\documentclass[12pt]{article}%
\usepackage{amssymb}
\usepackage{amsfonts}
\usepackage{amsmath}
\usepackage{array}
\usepackage{bm}
\usepackage{xcolor}
\usepackage[nohead]{geometry}
\usepackage[bottom]{footmisc}
\usepackage{endnotes}
\usepackage{graphicx}
\usepackage{multirow}
\usepackage{rotating}
\usepackage{setspace}
\usepackage{comment}
\usepackage{subfig}
\usepackage{lineno}
\usepackage{booktabs}
\usepackage[round, comma, sort]{natbib}
\usepackage{etoolbox}
\usepackage[colorlinks=true, citecolor=cyan, linkcolor= blue]{hyperref}
\newtheorem{theorem}{Theorem}

\newtheorem{corollary}{Corollary}

\newtheorem{lemma}{Lemma}
\newtheorem{proposition}{Proposition}
\newenvironment{proof}[1][Proof]{\noindent\textbf{#1.} }{\ \rule{0.5em}{0.5em}}
\makeatletter
\def\@biblabel#1{\hspace*{-\labelsep}}
\patchcmd{\NAT@citex}
  {\@citea\NAT@hyper@{%
     \NAT@nmfmt{\NAT@nm}%
     \hyper@natlinkbreak{\NAT@aysep\NAT@spacechar}{\@citeb\@extra@b@citeb}%
     \NAT@date}}
  {\@citea\NAT@nmfmt{\NAT@nm}%
   \NAT@aysep\NAT@spacechar\NAT@hyper@{\NAT@date}}{}{}
\patchcmd{\NAT@citex}
  {\@citea\NAT@hyper@{%
     \NAT@nmfmt{\NAT@nm}%
     \hyper@natlinkbreak{\NAT@spacechar\NAT@@open\if*#1*\else#1\NAT@spacechar\fi}%
       {\@citeb\@extra@b@citeb}%
     \NAT@date}}
  {\@citea\NAT@nmfmt{\NAT@nm}%
   \NAT@spacechar\NAT@@open\if*#1*\else#1\NAT@spacechar\fi\NAT@hyper@{\NAT@date}}
  {}{}
\makeatother
\begin{document}
\title{%
Information Design for Adaptive Organizations%
\thanks{%
The author is grateful to Makoto Hanazono and Keiichi Kawai for their valuable comments and suggestions. Financial support from the Japan Society for the Promotion of Science (JSPS), under Grant Numbers 16K17078, 20K20756, and 24K04780, is gratefully acknowledged.}
}
\author{%
Wataru Tamura%
\thanks{%
Graduate School of Economics, Nagoya University, Furo-cho, Chikusa-ku, Nagoya 464-8601, Japan. E-mail: wtr.tamura@gmail.com.
}%
}
\maketitle

\onehalfspacing
\begin{abstract}\noindent
This paper examines the optimal design of information sharing in organizations. Organizational performance depends on agents adapting to uncertain external environments while coordinating their actions, where coordination incentives and synergies are modeled as graphs (networks). The equilibrium strategies and the principal's objective function are summarized using Laplacian matrices of these graphs. I formulate a Bayesian persuasion problem to determine the optimal public signal and show that it comprises a set of statistics on local states, necessarily including their average, which serves as the organizational goal. When the principal benefits equally from the coordination of any two agents, the choice of disclosed statistics is based on the Laplacian eigenvectors and eigenvalues of the incentive graph. The algebraic connectivity (the second smallest Laplacian eigenvalue) determines the condition for full revelation, while the Laplacian spectral radius (the largest Laplacian eigenvalue) establishes the condition for minimum transparency, where only the average state is disclosed.
\vspace{4mm}\vspace{4mm}
\end{abstract}

\clearpage
\section{Introduction}

In many organizations, the principal possesses valuable information that can be strategically used to enhance organizational performance. For instance, headquarters gather and analyze marketing data and sales records to support product development and business operations. Similarly, managers in workplaces often leverage their market experience or knowledge of evaluation policies to guide employees effectively. Political parties also rely on sharing major policy lines with members to ensure consistent campaigning efforts.

When incentive mechanisms such as contracts or promotion systems successfully align agents' incentives with the organization's objectives, there is little need to restrict the information provided to agents. However, in cases where incentives are not fully aligned, unrestricted information sharing can increase losses from misaligned incentives. The principal may need to design information-sharing policies to mitigate these losses.

In this paper, I address the question of what information about the state of the world should be shared with agents in an organization. To explore this, I adopt a Bayesian persuasion framework in which the principal (e.g., an information-producing division) determines, ex ante, what information about the state to acquire and disclose to agents. The choice of information is represented by a public signal---a mapping from the state to a distribution over signal realizations. Agents, who are initially uninformed about the state, choose their actions based on the observed signal realizations.

The analysis focuses on a setup where organizational performance depends on two key factors: how well agents adapt to an uncertain external environment and how effectively they coordinate their actions. The external environment is represented by a set of local states, corresponding to the optimal actions if coordination were not a consideration. Coordination introduces a tradeoff: while it can enhance overall performance, it may also limit agents' responsiveness to external changes. A key parameter in the analysis is the importance of coordination relative to adaptation, which directly influences the optimal information-sharing policy.

This paper employs graph theory to model the coordination structure within organizations. Two central concepts are introduced: the \textit{synergy graph}, which represents the coordination that benefits the organization as a whole, and the \textit{incentive graph}, which reflects the coordination incentives naturally present among agents. The differences between these graphs illustrate the misalignment of incentives, providing a measure of the effectiveness and limitation of the organization's underlying incentive schemes. Using this framework, I explore how the coordination structure influences the design of optimal information-sharing policies.  

In the present model, the principal's objective function, given the agents' equilibrium strategies, is characterized using the Laplacian matrices of the synergy and incentive graphs. Leveraging the tractable properties of these matrices, I derive several insights into the design of the optimal public signal, which involves the disclosure of multiple statistics about local states. First, the optimal signal always informs agents of the average state, regardless of the coordination structure or the relative importance of coordination and adaptation. This finding aligns with the intuition that sharing the organization's primary goal is essential for fostering alignment among its members. Second, I identify two sufficient conditions under which full revelation is optimal. These conditions indicate that withholding information is detrimental when adaptation is more critical than coordination or when the incentive graph includes more links/edges than the synergy graph.

Further analysis focuses on the case where the principal benefits equally from coordination between any two agents, modeled by the synergy graph being a complete graph. In this setting, the optimal signal is determined through a two-step procedure. In the first step, statistics on the local states are computed using the Laplacian eigenvectors of the incentive graph. In the second step, the statistics for disclosure are selected based on whether the importance of coordination exceeds the cutoffs determined by the corresponding Laplacian eigenvalues. 

The condition for full revelation is determined by the second smallest Laplacian eigenvalue (the \textit{algebraic connectivity}) of the incentive graph, while the condition for minimum transparency is determined by the largest Laplacian eigenvalue (the \textit{Laplacian spectral radius}) of the incentive graph. As coordination becomes more critical for the organization, the principal reduces the number of disclosed statistics, thereby inducing stronger correlations in agents' equilibrium actions. In the extreme case, the gain from the optimal signal converges to the gain achieved by disclosing only the average state. 

Finally, I examine the implications of these findings using random graph models. Section \ref{sec:extension} focuses on the dimensionality of the optimal signal and investigates how network structures---such as connectivity and the presence of hubs---affect the principal's ability to balance adaptation and coordination. These numerical analyses provide further insights into how network topology shapes the design of information-sharing policies, particularly addressing the conditions under which minimum transparency is optimal.

\paragraph{Related literature}

In organizational economics, a central theme is the tradeoff between adaptation and coordination, often hindered by underlying incentive misalignments that prevent efficient information sharing. Using tractable formulations with quadratic specifications as in the present paper, various aspects have been examined, including decentralization (\citealp{alonso2008does}; \citealp{ran:2008res}), specialization (\citealp{ds:2006jpe}), leadership traits (\citealp{bolton2013leadership}), and managerial attention (\citealp{dessein2021managerial}).

Along these lines, \cite{herskovic2020acquiring} study information externalities in endogenously formed networks, showing that equilibrium networks exhibit hierarchical structures. \cite{calvo2009information} investigate the impact of exogenous communication networks within a beauty-contest model. \cite{galeotti2013strategic} and \cite{calvo2015communication} consider frictional communication and examine the equilibrium configuration of information flows as endogenous communication networks. Note that \cite{calvo2015communication} also introduce (weighted and directed) networks describing agents' coordination incentives to examine the interplay between the coordination payoff structure and the equilibrium information flows.%
\footnote{With a similar specification, \cite{hk:2016el} examine strategic communication of hard information in networks and provide a sufficient condition for a full-revelation equilibrium.}
These studies show that the equilibrium information flows in decentralized communication depend considerably on the specification of private information and communication protocols. In contrast, this paper focuses on centralized communication, where the principal provides information publicly to the agents, avoiding potential complications from decentralized information flows that arise when private communication is used.

The literature on Bayesian persuasion has provided valuable insights into information design for organizations. For instance, \cite{jeh:2015res} develop a general framework for Bayesian persuasion with multidimensional states, demonstrating that full disclosure is typically suboptimal. While their approach is highly general, this paper adopts a quadratic specification to derive explicit properties of the optimal signal. Similarly, \cite{iva:2010jet} examine informational control in the context of organizational design, focusing on how the principal can restrict the information available to the agent to align actions more closely with organizational objectives and improve bottom-up communication.

The inducement of correlation in agents' expectations has been explored in the literature on quadratic Bayesian persuasion. \cite{rs:2010jpe} analyze a setting where the principal seeks to maximize the covariance of agents' expectations, using randomized signals in a discrete state space. \cite{tamura2018bayesian} extend this framework to Gaussian state distributions and introduce the concept of an ``informational budget,'' illustrating that inducing correlations inherently reduces overall informativeness. This paper builds on these insights by investigating how agents' positions within networks, represented by the incentive and synergy graphs, shape the optimal signal to balance informativeness for adaptation and correlation in expectations to facilitate coordination in organizational contexts.

This paper also contributes to the growing intersection of public information design and graph theory. \cite{candogan2022persuasion} analyzes persuasion in networks with binary actions and local payoff complementarities, introducing optimal public signaling mechanisms characterized through convex optimization. The study examines how network properties, such as cores, shape equilibrium outcomes and demonstrates the effectiveness of asymptotically optimal mechanisms in large random networks. In contrast, this paper models coordination structures using Laplacian matrices, offering a novel perspective on how the alignment or misalignment of incentives influences information design. The use of spectral properties, such as algebraic connectivity and Laplacian spectral radius, provides new tools for understanding the interplay between network topology, public signals, and organizational performance.

While less directly related, \cite{egorov2020persuasion} study persuasion in networks where agents decide whether to subscribe to a sender's signal or rely on their neighbors for information. Unlike the present paper, which examines exogenous networks representing payoff externalities, their analysis focuses on the dynamics of information diffusion through network links without payoff externalities.%
\footnote{Recent works in this line include \cite{candogan2020information}, \cite{galperti2023games}, \cite{kerman2023pitfalls}, \cite{chen2024optimality}, and \cite{haghtalab2024leakage}.}
 
 Network-based intervention policies are also studied in the literature, building on the network-game model introduced by \cite{ballester2006s}, where equilibrium actions are linked to the Katz-Bonacich centrality of the adjacency matrix.%
 \footnote{Using this quadratic setup, \cite{candogan2012optimal} and \cite{bloch2013pricing} examine optimal monopoly pricing in networks with consumption externalities. They characterize optimal discriminatory pricing or induced demands in terms of Katz-Bonacich centrality, without directly relating their results to graph spectral properties.}
 Notably, \cite{galeotti2020targeting} apply spectral graph theory to study optimal targeting of interventions in networks. They show that the optimal policy depends on the principal components of the adjacency matrix: for strategic complements, it aligns with the eigenvector associated with the largest eigenvalue, while for strategic substitutes, it aligns with the eigenvector associated with the smallest eigenvalue, thereby internalizing positive or negative externalities. While their focus is on incentive manipulation using the adjacency matrix, this paper examines public signals based on the spectral decomposition of the Laplacian matrix, highlighting the role of intermediate eigenvalues and corresponding eigenvectors in shaping information-intervention policies.

Furthermore, \cite{shimono2025lec} extend the framework developed here, using the informativeness of optimal public signals to propose the Laplacian Eigenvector Centrality (LEC) as a novel centrality measure in network theory. Intuitively, as highlighted in Subsection \ref{sec:example} below, highly connected `central' agents within the network play a crucial role in facilitating coordination across the entire organization. Consequently, the optimal signal provides more precise information about the local states of these agents, thereby effectively balancing adaptation and coordination. From a network analysis perspective, the solution presented in this paper can be used to identify and quantify the centrality of nodes within the network (see Subsection \ref{sec:informativeness} for further discussion).

In summary, this paper advances the literature by integrating spectral graph theory into the analysis of information design in organizations. It provides explicit conditions for optimal public signals and highlights the role of network topology in shaping information-sharing policies. By focusing on robust public communication, the paper complements existing studies on decentralized information flows and expands the theoretical framework for adaptive organizational design.

\section{Model}
An organization consists of one principal and $n$ agents. Each agent chooses an action $a_{i} \in \mathbb{R}$, where the action profile is denoted by $\bm{a} = (a_{1}, \dots, a_{n})' \in \mathbb{R}^{n}$.%
\footnote{Throughout this paper, vectors are interpreted as column vectors.}
Organizational activities depend on two key factors: adaptation and coordination. Adaptation refers to how well actions are adjusted to uncertain external environments, represented by a vector of local states $\bm{x} = (x_{1}, \dots, x_{n})' \in \mathbb{R}^{n}$. Coordination refers to how well the actions are aligned. 

Let $\bm{1}_n \equiv (1,\dots, 1)'$ denote the vector of ones and $I_n$ denote the identity matrix of size $n$. Suppose that $\bm{x}$ is distributed according to $N(\mu \bm{1}_n, \sigma^{2}I_n)$, where $\mu \in \mathbb{R}$ and $\sigma^2 > 0$.%
\footnote{In Section \ref{sec:correlation}, an extension is examined where the local states are correlated.}
It is also assumed that agents are ex ante uninformed about the state and receive a public signal before choosing their actions.

The payoff function of agent $i$ is defined by
\begin{equation}
u_{i}(\bm{a},\bm{x}) = -(a_{i} - x_{i})^{2} - \beta \sum_{j=1}^{n} g_{ij} (a_{i} - a_{j})^{2},	
\end{equation}
and the objective function of the principal is
\begin{equation}
v(\bm{a}, \bm{x}) = -\sum_{i=1}^{n} (a_{i} - x_{i})^{2} - \beta \sum_{i=1}^{n} \sum_{j=1}^{n} \tilde{g}_{ij} (a_{i} - a_{j})^{2}.	
\end{equation}
The parameter $\beta \geq 0$ measures the importance of coordination relative to adaptation.%
\footnote{In Section \ref{sec:asymmetry}, the case in which the principal and agents have different degrees of coordination benefits is examined.}
The organizational structure regarding coordination is summarized by two symmetric adjacency matrices $G = [g_{ij}]$ and $\widetilde{G} = [\tilde{g}_{ij}]$, where $g_{ij}, \tilde{g}_{ij} \in \{ 0, 1 \}$ for all $i$ and $j$.%
\footnote{All propositions presented below still hold in an extension where the organizational structure is represented by two weighted graphs such that $g_{ij}, \tilde{g}_{ij} \in [ 0, 1 ]$ for all $i$ and $j$.}
The diagonal entries in $G$ and $\widetilde{G}$ are all zero ($g_{ii} = \tilde{g}_{ii} = 0$ for every $i$).
In this context, $G$ represents the \textit{incentive graph}, capturing the coordination incentives between agents, and $\widetilde{G}$ represents the \textit{synergy graph}, which specifies the desired coordination pattern from the principal's perspective.

Two comments on incentive distortion are appropriate here.  
First, $G \neq \widetilde{G}$ implies that private incentives for coordination are not aligned with the organizational objective.  
For example, when $(g_{ij}, \tilde{g}_{ij}) = (1, 0)$, agents $i$ and $j$ have incentives to coordinate, but their coordination does not benefit the organization. Their coordination incentives may even be detrimental in equilibrium because they could crowd out the incentives to adapt to the external environment. Conversely, when $(g_{ij}, \tilde{g}_{ij}) = (0, 1)$, the principal benefits from coordination between agents $i$ and $j$, but the agents have no incentive to coordinate. In such cases, the principal must induce correlation in actions by creating a correlation in expectations about $x_{i}$ and $x_{j}$.  
Second, agent $i$ does not fully internalize the coordination loss $(a_{i} - a_{j})^{2}$ incurred by agent $j$. Thus, the equilibrium degree of coordination is lower than the efficient degree of coordination, even when there is no distortion in the coordination structure (i.e., $G = \widetilde{G}$).

The principal determines the information available to agents by choosing a \textit{public signal} $\varphi: \mathbb{R}^{n} \to \Delta (M)$, where the set of signal realizations $M$ can be any measurable set. A signal $\varphi$ is said to \textit{inform the agents of} variable $y$ (denoted by $y \in \mathcal{I}(\varphi)$) if $\mathbb{E}[y|m] = y$ for all $m \in M$. As noted by \cite{kamenica2011bayesian}, the choice of signal can be interpreted as the institutional design for gathering, processing, and disclosing information. The purpose of this paper is to identify the signal that maximizes the expected payoff of the principal, given the agents' equilibrium strategies.

\section{Preliminaries}
This section presents preliminary results for analyzing the optimal public signal. Subsection \ref{sec:laplacian} introduces basic facts about spectral graph theory, which are used to derive the equilibrium strategies in Subsection \ref{sec:equilibrium}. Subsection \ref{sec:optimal} formulates the principal's optimization problem and presents the optimal public signal based on the existing literature.

\subsection{The Laplacian spectrum of graphs}\label{sec:laplacian}
Let $L = [L_{ij}]$ be the Laplacian matrix of $G$, defined by
\begin{equation}
L_{ij} =
	\begin{aligned}
\begin{cases}
\sum_{k=1}^{n} g_{ik} & \text{if } i = j, \\
-g_{ij} & \text{if } i \neq j.
\end{cases}		
	\end{aligned}
\end{equation}
That is, $L = D_{G} - G$, where $D_{G}$ is the degree matrix of $G$ (i.e., a diagonal matrix whose $i$th diagonal entry is the degree of vertex $i$).%
\footnote{See Appendix \ref{sec:laplacian_a} for the basic properties of the Laplacian spectrum used in this paper.}
Similarly, $\widetilde{L} = D_{\widetilde{G}} - \widetilde{G}$ denotes the Laplacian matrix of $\widetilde{G}$.
A graph is said to be complete, denoted by $K_n$, if all vertices are adjacent to one another.  
Let $L^{K}$ be the Laplacian of a complete graph. The eigenvalues of $L^{K}$ are given by
\begin{equation}
0 = \lambda_{1}^{K} < \lambda_{2}^{K} = \cdots = \lambda_{n}^{K} = n.
\end{equation}
The vector of ones $\bm{q}_{1}^{K} = \bm{1}_{n} \equiv (1, \dots, 1)'$ is an eigenvector of $L^{K}$ associated with the eigenvalue $\lambda_{1}^{K} = 0$.  
Any vector $\bm{q}_{j}^{K}$ orthogonal to $\bm{1}_{n}$ (i.e., $\bm{q}_{j}^{K} \cdot \bm{1}_n = 0$) is an eigenvector of $L^{K}$ with eigenvalue $\lambda_{j}^{K} = n$.  

For any symmetric adjacency matrix $G$, the eigenvalues of its Laplacian matrix are expressed in ascending order as
\begin{equation}
0 = \lambda_{1} \leq \lambda_{2} \leq \cdots \leq \lambda_{n} \leq n.	
\end{equation}
Let $\bm{q}_{j}$ be the eigenvector of $L$ associated with the eigenvalue $\lambda_{j}$.  
As with the complete graph, $\bm{q}_{1} = \bm{1}_{n}$ is an eigenvector associated with $\lambda_{1} = 0$.  

In spectral graph theory, it is established that the second smallest eigenvalue $\lambda_{2}$, called the \textit{algebraic connectivity}, equals 0 if and only if $G$ is not connected.  
Similarly, the largest eigenvalue $\lambda_{n}$, called the \textit{Laplacian spectral radius}, equals $n$ if and only if the complement of $G$ is not connected.

\subsection{Equilibrium and efficient strategies}\label{sec:equilibrium}
Each agent maximizes their expected payoff given the common posterior induced by the public signal.  
Let $\hat{\bm{x}} \equiv \mathbb{E}[\bm{x}|m]$ denote the public expectation of the state.  
The unique Nash equilibrium is given by a linear function of $\hat{\bm{x}}$:%
\footnote{For the derivation, see Appendix \ref{sec:derivation_equil}.}
\begin{equation}
	\bm{a}^{*} = (I_n + \beta L)^{-1} \hat{\bm{x}}.
\label{eq:ne}
\end{equation}
Let $B = (I_n + \beta L)^{-1}$ denote the matrix that describes how agents' expectations of the state determine their equilibrium actions, reflecting the coordination incentives dictated by $L$.

On the other hand, the unique efficient solution, which optimizes the total outcome for the organization as a whole, is given by:
\begin{equation}
\bm{a}^{**} = (I_n + 2 \beta \widetilde{L})^{-1} \hat{\bm{x}}.
\label{eq:ts}	
\end{equation}

Two types of distortions arise from the differences between the equilibrium strategy in \eqref{eq:ne} and the efficient strategy in \eqref{eq:ts}. First, the distortion in the coordination structure is represented by the difference between $L$ and $\widetilde{L}$, reflecting discrepancies between how coordination is implemented and how it is ideally structured. Second, the distortion in the degree of coordination is captured by the difference between $\beta$ and $2\beta$, indicating that agents do not fully internalize the benefits of coordination, which can lead to suboptimal collective outcomes.

\subsection{Optimal public signal characterization}\label{sec:optimal}
Given the unique Nash equilibrium, the principal's objective function can be expressed as:%
\footnote{For the derivation, see Appendix \ref{sec:derivation_objective}.}%
\begin{equation}\label{eq:objective}
	\mathbb{E}[ v(a^{*}, x)]
= -\mathbb{E}[x'x] + \mathbb{E}[\hat{\bm{x}}' V\hat{\bm{x}}],
\end{equation}
where $V$ is a symmetric matrix that reflects the strategic effects of coordination benefits and structures:
\begin{equation}\label{eq:payoff_V}
	V \equiv B'(I_n - 2 \beta (\widetilde{L} - L))B.
\end{equation}
This matrix captures the trade-offs in information design, particularly in evaluating whether inducing specific correlations between $\hat{\bm{x}}_i$ and $\hat{\bm{x}}_j$ benefits the organization, as discussed in the next section.

The optimal public signal $\varphi^{*}$ is determined by solving the following optimization problem:%
\begin{equation}
\max_{\varphi}~\mathbb{E}[\hat{\bm{x}}' V \hat{\bm{x}}].
\end{equation}
A key result from the quadratic Bayesian persuasion problem is stated in the following theorem:
\begin{theorem}[\citealp{tamura2018bayesian}]\label{theorem}
Suppose that each $x_{i}$ is independent and identically distributed according to a normal distribution.
Let $\omega_j$ denote an eigenvalue of $V$ and let $\bm{z}_j$ denote the corresponding eigenvector.
Assume that $V$ has $r$ nonnegative eigenvalues and $n-r$ negative eigenvalues.
Then, the optimal public signal consists of $r$ statistics $\bm{m}= \{ m_{j}^{*} \} $, where each statistic $m_j^{*}$ is a linear combination of the state, given by $m_{j}^{*} = \bm{z}_{j}' \bm{x}$ for $j$ such that $\omega_j \geq 0$.
\end{theorem}

This theorem highlights how the eigenstructure of $V$ determines the optimal design of the public signal. The principal focuses on inducing correlations aligned with the positive eigenvalues of $V$ while avoiding those linked to negative eigenvalues. The subsequent sections of this paper examine how the solutions to the information design problem depend on the coordination benefit $\beta$ and the differences in the coordination structures $(G, \widetilde{G})$.

\section{General results}
This section outlines key findings on how information is optimally shared within organizations, regardless of the specific coordination structures or incentive alignments.%
\footnote{For all proofs, see Appendix \ref{sec:proof}.}

\subsection{Average state as a fundamental information component}

The first result highlights a universal principle in organizational information strategy:
\begin{proposition}\label{prop_avg}
For any coordination structure $(G, \widetilde{G})$ and level of coordination benefits $\beta \geq 0$, the optimal signal must inform the agents of the average state:
\begin{equation}
\bar{x} \equiv \frac{1}{n}\sum_{i=1}^{n} x_{i} \in \mathcal{I}(\varphi^{*}).
\end{equation}
\end{proposition}

This result demonstrates that, regardless of complexities such as the coordination structures, the principal must ensure that the agents are informed about the overall goal of the organization. Proposition \ref{prop_avg} also implies that no revelation is suboptimal. Beyond the average state, what other types of statistics should the principal provide?

\subsection{Conditions for full revelation optimality}
This subsection establishes the sufficient conditions under which the provision of full information about the state is optimal.%
\footnote{Note that Propositions \ref{prop_full} and \ref{prop_full_beta} are independent of the prior distribution of the state.}
Technically, the condition is simply whether matrix $V$ is positive semidefinite (i.e., whether all eigenvalues of $V$ are nonnegative).

The first proposition examines the condition on the coordination structure:
\begin{proposition}\label{prop_full}
If the incentive graph is a spanning supergraph of the synergy graph {\em (}i.e., if $g_{ij} \geq \tilde{g}_{ij}$ for all $i$ and $j${\em )}, then full revelation is optimal.%
\footnote{A spanning supergraph is a graph obtained by adding edges to a given graph while preserving the vertex set.}
\end{proposition}

This result reflects the inherent mechanics of the quadratic Bayesian persuasion: the principal coordinates agents' actions by withholding information, necessarily reducing their variability. When the coordination incentives in the incentive graph exceed those in the synergy graph (i.e., $g_{ij} \geq \tilde{g}_{ij}$ for all $i$ and $j$), agents already have strong incentives to align their actions and place lower weights on adaptation. In such cases, the principal's focus shifts to enhancing adaptation by ensuring that agents' actions respond effectively to the state of the world. Withholding information fails to improve outcomes here, as there is no need to mitigate excessive coordination. Instead, the challenge lies in providing agents with precise information to guide their actions. As a result, full transparency becomes the optimal strategy, allowing agents to adapt effectively to the local states while taking advantage of their existing coordination incentives.

To illustrate this more concretely, consider a two-agent organization where $(g_{12}, \tilde{g}_{12}) = (1, 0)$. In this case, the equilibrium strategy is
\begin{equation}\nonumber
a_i^{*} = \frac{1 + \beta}{1 + 2\beta} \hat{x}_i + \frac{\beta}{1 + 2\beta} \hat{x}_{-i},
\end{equation}
while the efficient solution is $a^{**}_i = \hat{x}_i$.
As $\beta$ increases from a low value to infinity, the equilibrium strategy shifts from relying on individual expectations $\hat{x}_i$ to averaging the expectations, $(\hat{x}_1 + \hat{x}_2)/2$.  

To highlight the costs and benefits of providing additional information beyond the average state $\bar{x} = (x_{1} + x_{2})/2$, consider a noisy signal $s = x_{1} + \epsilon$, where $\epsilon \sim N(0, \kappa^{2})$.  
When the signal informs the agents of the realizations of $s$ and $\bar{x}$, the principal's expected payoff, in the setting where $(g_{12}, \tilde{g}_{12}) = (1, 0)$, is expressed as
\begin{equation}\nonumber
\begin{aligned}
\mathbb{E}\left[ v(\bm{a}^{*}, \bm{x}) \right] =&
-\mathbb{E}\left[ \sum_{i} (a_{i}^{*} - x_i)^{2} \right] \\
=& -\mathbb{E}\left[ 2\left( \frac{\beta}{1 + 2\beta} \right)^{2} \left( \hat{x}_{1} - \hat{x}_{2} \right)^{2} + \sum_{i} (\hat{x}_{i} - x_{i})^{2} \right] \\
=& -\sigma^{2} \left[ \left( \frac{2\beta}{1 + 2\beta} \right)^{2} \frac{\sigma^{2}}{\sigma^{2} + 2\kappa^{2}} + \frac{2\kappa^{2}}{\sigma^{2} + 2\kappa^{2}} \right].
\end{aligned}
\end{equation}
Here, the first term inside the brackets represents the variance due to misaligned coordination incentives, while the second term captures the variance due to estimation error. As $\kappa$ decreases (i.e., the precision of $s$ increases), the reduction in estimation error outweighs the increase in variance due to misaligned coordination. Since $2\beta / (1 + 2\beta) < 1$, the principal's overall payoff necessarily increases as $\kappa$ decreases, demonstrating that full revelation is indeed beneficial in this context.

Returning to the general results, an immediate corollary of Proposition \ref{prop_full} is as follows:
\begin{corollary}
If $G = K_n$, where $K_n$ denotes a complete graph, then full revelation is optimal.
\end{corollary}

The next proposition provides another sufficient condition on $\beta$, independent of $(L, \widetilde{L})$:
\begin{proposition}\label{prop_full_beta}
If $\beta \leq 1/(2n)$, then full revelation is optimal.
\end{proposition}

When coordination benefits are sufficiently small, the distortion from the coordination structure becomes negligible. In such cases, the principal should focus on enhancing adaptation by providing precise information about each local state.

For instance, consider again a two-agent organization, but now assume that $(g_{12}, \tilde{g}_{12}) = (0, 1)$.  
In this case, since the agents have no incentive to coordinate, each agent's strategy becomes $a^{*}_{i} = \hat{x}_{i}$.  
When the signal consists of $m = (s, \bar{x})$, where $s = x_{1} + \epsilon$ and $\bar{x} = (x_{1} + x_{2})/2$, as in the previous example, the principal's expected payoff becomes
\begin{equation}
\begin{aligned}\nonumber
\mathbb{E}[v(\bm{a}^{*}, \bm{x})] =&
-\mathbb{E}\left[ \sum_{i} (\hat{x}_{i} - x_{i})^{2} + 4\beta \left( \hat{x}_{1} - \hat{x}_{2} \right)^{2} \right] \\
=& -\sigma^{2} \left[ \frac{2\kappa^{2}}{\sigma^{2} + 2\kappa^{2}}
+ 4\beta \frac{\sigma^{2}}{\sigma^{2} + 2\kappa^{2}} \right].
\end{aligned}
\end{equation}

The effect of increased transparency depends on whether $\beta$ exceeds $1/4$. When $\beta < 1/4$, reducing the noise in $s$ (decreasing $\kappa$) benefits the principal, as it reduces estimation error more significantly than it increases distortion from misaligned coordination incentives. Conversely, if $\beta > 1/4$, the loss from misaligned coordination outweighs the benefits of reduced noise, suggesting that the principal might be better off refraining from revealing more information beyond $\bar{x}$.

\subsection{Determining the dimension of the optimal public signal}

This subsection investigates how changes in the underlying incentive structure affect the dimension of the optimal signal, i.e., the number of statistics to be disclosed. Specifically, I analyze how the number of nonnegative eigenvalues of $V$ varies with factors such as the coordination benefit $\beta$ and the incentive graph $G$.

First, the following proposition provides a clear comparative static result regarding the impact of the incentive graph $G$:
\begin{proposition}\label{prop_morelinks}
Fix $\beta$ and $\widetilde{G}$. If $G'$ is a spanning supergraph of $G$ {\em (}i.e., \( g'_{ij} \geq g_{ij} \) for all \( i \) and \( j \){\em )}, then the dimension of the optimal public signal is not smaller when the incentive graph is \( G' \) rather than \( G \).
\end{proposition}

This result aligns with the findings in the previous subsection: the principal benefits from restricting information sharing only when agents have weaker incentives to coordinate their actions with each other. Notably, adding a link between $i$ and $j$ in $G$ increases the signal dimensionality, even when the principal does not benefit from such coordination (i.e., $\tilde{g}_{ij} = 0$). This highlights that aligning coordination incentives between the principal and agents does not necessarily enhance transparency.

The next proposition examines the impact of $\beta$ on the structure of the optimal public signal, specifically addressing how it determines the number of disclosed statistics:
\begin{proposition}\label{prop_beta}
The dimension of the optimal public signal is nonincreasing as $\beta$ increases.
\end{proposition}

In summary, this section outlines key conditions for the optimal design of public signals, highlighting when full information revelation is beneficial and how the coordination structure and coordination importance influence the signal dimensionality. The next section focuses on the specific case where the synergy graph is a complete graph ($\widetilde{G} = K_n$), providing deeper insights into the relationship between coordination and information design.

\section{Optimal signal for uniform synergy graphs}

This section focuses on environments where the principal benefits from coordination among all pairs of agents. Formally, suppose that $\widetilde{G} = K_n$ (i.e., $\tilde{g}_{ij} = 1$ for all \(i \neq j\)). This setting represents scenarios where achieving uniform coordination across the entire organization is crucial for maximizing overall performance. Such cases are particularly relevant in highly interconnected systems, such as teams working on collaborative projects or supply chains requiring seamless synchronization among all entities. Under this assumption, I examine how the organizational parameters \((G, \beta)\)---the agents' incentive graph and coordination importance---determine the statistics to be shared with agents.

\subsection{Characterization}

When $\widetilde{G} = K_n$, Theorem \ref{theorem} provides a tractable characterization of the optimal statistics based on the eigenvectors and eigenvalues of the Laplacian matrix of $G$.

\begin{proposition}\label{prop_balanced}
Suppose that $\widetilde{G} = K_n$. Define $m_{j}^{*} \equiv \bm{q}_{j}' \bm{x}$ as the linear combination of the local states based on the eigenvector $\bm{q}_{j}$ of $L$. The optimal signal informs the agents of a subset of $\{ m_{1}^{*}, \dots, m_{n}^{*} \}$, and the elements of the optimal signal are determined by the coordination benefit as follows: for $j = 2, \dots, n$,
\begin{equation}\label{eq:inform}
m_{j}^{*} \in \mathcal{I}(\varphi^{*}) \iff \beta^{-1} \geq 2(n - \lambda_{j}).
\end{equation}
\end{proposition}

According to Proposition \ref{prop_balanced}, the Laplacian eigenvectors are used to compute the statistics, while the Laplacian eigenvalues determine which statistics are disclosed. Each statistic $m_{j}^{*}$ is disclosed when $\beta$ is below the threshold $\widehat{\beta}(\lambda_{j}) \equiv 1/(2(n - \lambda_{j}))$.%
\footnote{For convenience, I set $\widehat{\beta}(\lambda) = \infty$ if $\lambda = n$. Note that $\lambda_n = n$ if and only if the complement of $G$ is not connected. For example, if there is an agent with degree $n-1$ (connected to all other agents), then $\lambda_n = n$.}
As $\beta$ increases, the optimal signal changes discretely, consisting of fewer statistics.
Notably, the condition for full revelation depends solely on the algebraic connectivity $\lambda_{2}$.
\begin{corollary}\label{corollary_full_beta}
Suppose that $\widetilde{G} = K_n$. Then, full revelation is optimal if and only if $\beta \leq \widehat{\beta}(\lambda_{2})$.
\end{corollary}

The next proposition provides a sufficient condition under which the optimal signal reveals no information beyond the average state.

\begin{proposition}\label{prop_single_beta}
Suppose that $\widetilde{G} = K_n$. If the complement of $G$ is connected, the optimal signal informs the agents only of the average state for sufficiently high $\beta$.
\end{proposition}

From Proposition \ref{prop_balanced}, the optimal signal reveals only the average state when $\beta > \widehat{\beta}(\lambda_{n})$. Even when the complement of $G$ is not connected (i.e., $\lambda_{n} = n$ or $\widehat{\beta}(\lambda_{n}) = \infty$) and, consequently, the optimal signal consists of multiple statistics even for large $\beta$, the gain from information design approaches the value achieved by disclosing only the average state.

\begin{proposition}\label{prop_gain}
Suppose that $\widetilde{G} = K_n$. As $\beta$ approaches infinity, the gain from information design converges to $\sigma^{2}$, which equals the limit of the gain from disclosing the average state.
\end{proposition}

This subsection characterized the optimal signal when $\beta$ is either low ($\beta \leq \widehat{\beta}(\lambda_2)$) or high ($\beta > \widehat{\beta}(\lambda_n)$). However, when $\beta$ is at an intermediate level, the choice of statistics depends crucially on the topology of the incentive graph $G$. The next subsection provides examples to illustrate how the structure of $G$ influences the design of the optimal public signal.

\subsection{Examples}\label{sec:example}
This subsection illustrates the optimal signal identified in Proposition \ref{prop_balanced} using several examples. The Laplacian spectra of the examples are drawn from \cite{moh:1991book}.

\subsubsection{Star graphs}

A graph is said to be a star $S_{n}$ if a single vertex has a degree of $n-1$, while all other vertices have a degree of 1. This structure represents a centralized organization where one agent, the ``core agent,'' acts as a hub, indirectly connecting all other agents, who only coordinate through the core. Such a setup is typical of hierarchical networks, such as a team leader overseeing multiple members who do not interact directly. For instance, in a research lab, a principal investigator (core agent) may coordinate among researchers (periphery agents). However, the success of the project often depends on collaboration among researchers themselves, beyond their alignment with the principal investigator.

When the incentive graph corresponds to a star graph, the optimal signal takes one of two forms depending on the level of coordination importance:
\begin{proposition}\label{prop_star}
Suppose that $\widetilde{G} = K_n$ and $G = S_{n}$ such that $g_{1j} = 1$ and $g_{ij} = 0$ for all $i, j \neq 1$. If $\beta > \frac{1}{2(n-1)}$, the optimal signal informs the agents of the local state of the core agent $x_{1}$ and the average state of the other agents $\bar{x}_{-1} \equiv \frac{1}{n-1}\sum_{j=2}^{n}x_{j}$. Otherwise, full revelation is optimal.
\end{proposition}

As $\beta$ increases, the optimal signal shifts from full disclosure to partial disclosure. Since every agent has an incentive to coordinate with the core agent, the core agent facilitates coordination between all agents. Consequently, the optimal signal necessarily informs the agents of the core agent's state in addition to the overall average state.

\subsubsection{Ring graphs}

A ring graph $R_n$ is a connected regular graph where each node has a degree of 2. Node $i$ is adjacent to nodes $i-1$ and $i+1$.%
\footnote{Nodes $1$ and $n$ are adjacent to each other.}
This structure often represents urban neighborhoods where households coordinate with immediate neighbors on local issues such as utilities or security. However, the city administration (principal) benefits from uniform coordination among all households to achieve broader goals, such as waste management or disaster preparedness, aligning with a complete graph as the efficient structure.

The Laplacian spectrum of the ring graph for $j = 1, 2, \dots, n$ is given by:
\begin{equation}
\lambda_j = 2 - 2 \cos \left( \frac{2 \pi k_{(j)} }{n} \right),
\end{equation}
where $k_{(1)} = 0$, $k_{(2)} = k_{(3)} = 1$, and so on. For the final case, $k_{(n)} = (n-1)/2$ when $n$ is odd, and $k_{(n)} = k/2$ when $n$ is even. 

As $\beta^{-1}$ increases, the principal gradually adds two statistics to the publication at each step. In each phase, agents are informed symmetrically, meaning that the variance of $\hat{x}_i$ is identical for all $i$. This contrasts with the optimal policy for the star graph, where adjustments are asymmetric and exhibit a threshold-like response to the adaptation/coordination trade-off. By comparison, the ring graph implements symmetric and gradual adjustments in response to changes in $\beta$.

\subsubsection{Social networks with isolated cliques}

Suppose that the incentive graph is composed of two completely connected clusters, \textit{cliques}, that are isolated from each other, represented by the following adjacency matrix:
\begin{equation}\nonumber
G = \begin{bmatrix} K_m & O_{n-m} \\ O_{n-m} & K_{n-m} \end{bmatrix}.
\end{equation}
The upper-left block of $G$, $K_m$, represents the clique of size $m$, while the bottom-right block, $K_{n-m}$, represents the other clique of size $n-m$. The off-diagonal zero blocks, $O_{n-m}$, indicate the absence of connections between the two clusters. Assume $m \leq n-m$.

This structure arises in social networks where individuals belong to distinct, tightly-knit communities with no direct interaction between groups. For example, it could represent two departments within an organization that operate independently but are overseen by a central authority (the principal), who seeks to align their activities for organizational success.

The Laplacian spectrum of this graph is $(0^2, m^{m-1}, (n-m)^{n-m-1})$, where the superscripts indicate multiplicities. The optimal public signal involves four distinct phases: the minimum transparency policy, two intermediate phases, and full transparency.

\begin{description}
	\item[Phase 1 ($\beta \in [0,\widehat{\beta}(n-m)$):] Only the overall average is disclosed.
	\item[Phase 2 ($\beta \in [\widehat{\beta}(n-m), \widehat{\beta}(m))$):] 
	The principal reveals information only about the local states in the larger cluster, in addition to the overall average. To ensure coordination, the principal withholds information about the cluster-average of the larger cluster, ensuring that the disclosed statistics have a zero mean within that cluster.  
	\item[Phase 3 ($\beta \in [\widehat{\beta}(m), \widehat{\beta}(0))$):]
	 The principal discloses information about the smaller cluster, while keeping the agents uninformed about its cluster-average to preserve its zero mean.
	 \item[Phase 4 ($\beta \in [\widehat{\beta}(0), \infty)$):] 
	 The principal reveals the cluster averages alongside the Phase 3 statistics, thereby achieving full transparency.
\end{description}

The intuition behind the optimal policy is as follows. Agents naturally have strong incentives to coordinate within their respective clusters. Thus, the principal's primary focus is to induce coordination between the clusters. To achieve this, the principal avoids revealing the difference between the clusters until $\beta$ exceeds the threshold identified in Proposition \ref{prop_full_beta}. By doing so, the principal carefully balances adaptation and coordination incentives to maximize organizational performance.%
\footnote{This result aligns with spectral clustering methods used to identify cluster structures within a given network (\citealp{von2007tutorial}). The final statistic disclosed in Phase 4 is constructed using the Fiedler vector (i.e., the eigenvector corresponding to the algebraic connectivity), assigning positive weights to the nodes in one clique and negative weights to those in the other. By combining the information of the overall average, this statistic reveals the cluster averages.}

\subsection{Signal informativeness}\label{sec:informativeness}

Under the optimal signal presented in Proposition \ref{prop_balanced}, the variance-covariance matrix of the posterior expectations is given by 
\begin{equation}\label{eq:cov_optimal}
	\text{var}(\hat{\bm{x}}) = \sigma^2 \bar{Q}(\bar{Q}' \bar{Q})^{-1}\bar{Q}',
\end{equation}
where $\bar{Q} = \begin{bmatrix} \bm{q}_1&\bm{q}_n&\cdots &\bm{q}_r \end{bmatrix}$ is the $n \times r$ matrix containing the Laplacian eigenvectors whose corresponding eigenvalues (except for $\bm{q}_1$) satisfy \eqref{eq:inform}. The informativeness of the optimal signal regarding each local state $x_i$ is commonly quantified by the variability of conditional expectations. Specifically, the variances of the conditional expectations, i.e., the diagonal elements of \eqref{eq:cov_optimal}, are proportional to the sum of the squared eigenvector components.%
\footnote{Here, the eigenvectors are implicitly normalized to unit length, $\bm{q}_i' \bm{q}_i = 1$ for all $i$, with orthogonality $\bm{q}_i' \bm{q}_j = 0$ for $i \neq j$.}
\begin{proposition}\label{prop_informativeness}
Suppose that the optimal public signal discloses $m_j^* = \bm{q}_j' \bm{x}$ for all $j \in \mathcal{I}^*$, where $\mathcal{I}^* = \{ j: m_{j}^{*} \in \mathcal{I}(\varphi^{*}) \}$ denotes the index set of disclosed statistics.
Then, the signal's informativeness regarding the local state $x_i$ is given by 
\begin{equation}
\frac{\mathrm{var}(\hat{x}_i)}{\mathrm{var}(x_i)} = \sum_{j \in \mathcal{I}^*} \bigl[q_j(i)\bigr]^2.
\end{equation}
\end{proposition}

The two-step characterization in Proposition \ref{prop_balanced} ensures that the informativeness of the optimal signal for any local state never decreases as $\beta^{-1}$ increases. However, the effects of adding or removing links in the incentive graph are nontrivial, as they alter the network structure, thereby changing the Laplacian eigenvectors and their corresponding eigenvalues.%
\footnote{For example, in some cases, the variance of $\hat{x}_i$ may decrease even when agent $i$ gains an additional link to agent $j$ in the incentive graph.}

From a network analysis perspective, \cite{shimono2025lec} show that the variances of expectations under the optimal signal capture the centrality of nodes in the network. The optimal signal reveals more precise information about the local states of highly connected agents compared to less connected ones. They also provide properties applicable to this paper. For instance, the symmetry property states that if two agents have identical sets of neighbors in the incentive graph, the optimal signal provides the same informativeness about their local states.

\subsection{Plus-one policy}

This subsection considers a simple---albeit suboptimal---\emph{plus-one policy}, in which the principal \emph{publicly} discloses the local state of a single agent $x_i$ in addition to the average state $\bar{x}$.%
\footnote{The result in this subsection remains unchanged if the principal discloses only a single local state without including the average state.}
The goal is to determine which agent should be targeted to maximize the organizational payoff under this policy.

Suppose that the principal discloses $(x_i, \bar{x})$, where $x_i$ is the fully disclosed local state of agent $i$, and $\bar{x}$ is the average state. Let $S = \mathrm{var}(\hat{\bm{x}})$ denote the variance-covariance matrix of the posterior expectations under this policy. 
Using \eqref{eq:payoff_V}, the improvement in the organizational payoff compared to the minimal transparency policy is given by
\begin{equation}\label{eq:gain_plusone}
\mathrm{tr}(VS) - \sigma^2 = 
\frac{\sigma^2 \,n}{n-1} \sum_{j=2}^{n}\omega_j \,[q_j(i)]^2,
\end{equation}
where $q_j(i)$ is the $i$th component of the Laplacian eigenvector $\bm{q}_j$, and $\omega_j$ is the eigenvalue of $V$ corresponding to $\bm{q}_j$:
\begin{equation}\label{eq:gain_omega}
\omega_j = \frac{1 - 2\beta (n - \lambda_j)}{\bigl(1 + \beta \lambda_j\bigr)^2}.
\end{equation}
Note that these eigenvalues satisfy $\omega_n \geq \omega_{n-1} \geq \dots \geq \omega_2$.%
\footnote{The sign of $\omega_j$ is consistent with the condition in \eqref{eq:inform}.}

\begin{proposition}\label{prop_plusone}
Under the plus-one policy, the principal should taret the agent \(i\) that maximizes 
$\sum_{j=2}^{n} \omega_j \, [q_j(i)]^2$.
\end{proposition}

This proposition shows that the principal's payoff gain is a weighted average of the eigenvalues \(\{\omega_n, \omega_{n-1}, \dots, \omega_2\}\), with weights determined by the squared components of the corresponding Laplacian eigenvectors. Thus, the principal achieves higher organizational gains by targeting the agent whose squared eigenvector components $\{ \bigl[q_n(i)\bigr]^2, \bigl[q_{n-1}(i)\bigr]^2, \dots \}$ are larger for the dominant eigenvalues \(\{\omega_n, \omega_{n-1}, \dots\}\). 

This result aligns with the insights from Subsection~\ref{sec:informativeness}, where the principal constructs public information by focusing on agents occupying the most ``central'' positions in the network in terms of Laplacian eigenvector centrality (\citealp{shimono2025lec}).%
\footnote{Indeed, \eqref{eq:gain_plusone} generalizes the notion of Laplacian eigenvector centrality by allowing the weights on the eigenvectors to be real-valued rather than strictly binary.}

In summary, the plus-one policy provides a practical framework for evaluating how targeting individual agents for information disclosure can enhance organizational performance. While suboptimal compared to more sophisticated designs, it illustrates how spectral properties of the network---particularly the Laplacian eigenvectors and their corresponding eigenvalues---can guide the principal's information design to improve collective outcomes.

\section{Extensions}\label{sec:extension}
\subsection{Correlation in local states}\label{sec:correlation}
In the main part of this paper, I focused on the case where the state variables are independent and examined the optimal public signal using the eigenvectors and eigenvalues of $V$ (Theorem \ref{theorem}). When the state variables are correlated, the optimal signal generally depends on the distribution of the state. The following theorem generalizes Theorem \ref{theorem} to account for correlations in the state.

\begin{theorem}[\citealp{tamura2018bayesian}]\label{thm_general}
Suppose that $\bm{x}$ is normally distributed according to $N(\bm{\mu},\Sigma)$ with positive definite covariance matrix $\Sigma$.
Suppose that $\Sigma^{\frac{1}{2}}V\Sigma^{\frac{1}{2}}$ has $r$ nonnegative eigenvalues $\omega_{1}, \dots, \omega_{r}\geq 0$ and $n-r$ negative eigenvalues $\omega_{r+1}, \dots, \omega_{n} < 0$.
Then, the optimal public signal is given by
\begin{equation}
\widetilde{\bm{m}}^{*} 
=\begin{bmatrix} \widetilde{m}_{1}^{*} \\ \vdots \\ \widetilde{m}_{r}^{*} \end{bmatrix}
= Z_{+}' \Sigma^{-\frac{1}{2}} \bm{x},	
\end{equation}
where $Z_{+}=[\bm{z}_{1}, \dots, \bm{z}_{r}]$ is a $n \times r$ matrix that contains the eigenvectors associated with the nonnegative eigenvalues of $\Sigma^{\frac{1}{2}}V \Sigma^{\frac{1}{2}}$.
\end{theorem}

It is worth noting that whenever the variance matrix $\Sigma$ of the local states is nonsingular, $\Sigma^{\frac{1}{2}}V\Sigma^{\frac{1}{2}}$ retains the same number of positive, negative, and zero eigenvalues as $V$.%
\footnote{This property is known as Sylvester's law of inertia.}
This implies that changes in the state distribution do not affect the dimensionality of the optimal public signal (i.e., the number of statistics to be disclosed).

In this subsection, I show that all the results presented in this paper hold when the distribution of the local states exhibits symmetric correlation. Formally, suppose that for every $i \neq j$, ${\rm cov}(x_{i},x_{j})=\rho \sigma^{2}$. In other words, the distribution of the local states is symmetric if the variance-covariance matrix of the state is expressed as 
\begin{equation}
\Sigma= \sigma^{2} \begin{bmatrix} 
1 & \rho & \ldots & \rho \\
\rho & 1 & \ldots & \rho \\
\vdots & \vdots & \ddots & \vdots \\
\rho & \rho & \ldots & 1 
\end{bmatrix}=\sigma^{2}\left[ (1-\rho) I_{n} + \rho \mathbf{1}_{n}\mathbf{1}_{n}' \right].	
\end{equation}

I begin by examining the eigenvectors and eigenvalues of $\Sigma$.
\begin{lemma}\label{lem_eigen_rhosigma}
Suppose that $\Sigma = \sigma^{2}[(1-\rho) I_{n} + \rho \mathbf{1}_{n}\mathbf{1}_{n}'] $.
Then, $\Sigma \mathbf{1}_{n}=\sigma^{2}(1-\rho+n\rho)\mathbf{1}_{n}$, and for any vector $\bm{q}$ of length $n$ such that $\mathbf{1}_{n}' \bm{q}=0$, $\Sigma \bm{q}= \sigma^{2} (1-\rho ) \bm{q}$.
\end{lemma}

From this lemma, it follows directly that
\begin{equation}
\begin{aligned}
\Sigma^{\frac{1}{2}}\mathbf{1}_{n}
=\sigma \sqrt{1-\rho+n\rho}\, \mathbf{1}_{n} ,	
\end{aligned}
\end{equation}
and for any vector $\bm{q}$ orthogonal to $\mathbf{1}_{n}$ (i.e., $\mathbf{1}_{n}'q=0$), 
\begin{equation}
\begin{aligned}
\Sigma^{\frac{1}{2}}\bm{q}
=\sigma \sqrt{1-\rho}\, \bm{q}.
\end{aligned}
\end{equation}

\begin{lemma}\label{lem_eigen_rhov}
Suppose $V=(I_{n}+\beta L)^{-1}(I_{n}-2\beta (\widetilde{L}-L ))(I_{n}+\beta L)^{-1}$ and $\Sigma=\sigma^{2}[(1-\rho)I_{n}+\rho \mathbf{1}_{n}\mathbf{1}_{n}']$.
Then, every eigenvector $\bm{q}$ of $V$ is also an eigenvector of $\Sigma^{\frac{1}{2}}V\Sigma^{\frac{1}{2}}$.
\end{lemma}

\begin{proposition}\label{prop_correlation}
Suppose that $\Sigma =\sigma^{2}[ (1-\rho)I_{n}+\rho \mathbf{1}_{n} \mathbf{1}_{n}']$.
Then, the optimal signal is independent of $\rho \in (-1,1)$.
\end{proposition}

This subsection demonstrates that the results established in the main part of the paper are robust to symmetric correlations in the state variables. Even when the local states exhibit such correlations, the dimensionality and structure of the optimal signal remain unaffected, as shown in Proposition~\ref{prop_correlation}.

\subsection{Asymmetry in coordination weights}\label{sec:asymmetry}

In the baseline model, the parameter representing the importance of coordination, $\beta \geq 0$, is assumed to be common to both the principal and the agents. In this subsection, I extend the model to allow for asymmetry in coordination benefits, introducing $\tilde{\beta}$ to represent the principal's coordination benefit, which may differ from $\beta$.

Suppose that the principal's objective function in equilibrium (i.e., $\bm{a} = B \hat{\bm{x}}$) is given by:%
\begin{equation}
\mathbb{E}[ v(\bm{a}^{*}, \bm{x})]
= -\mathbb{E}[x'x] + 
\mathbb{E}[\hat{\bm{x}}' B (I - 2(\tilde{\beta} \widetilde{L} - \beta L)) B \hat{\bm{x}}].
\label{eq:objective_extend}
\end{equation}

In this extended model, Proposition \ref{prop_avg} remains valid: for any $\beta \geq 0$ and $\tilde{\beta} \geq 0$, the optimal signal always informs the agents of the average state. However, other propositions require slight modifications.

\begin{proposition}
Full revelation is optimal if {\em (i)} $\beta g_{ij} \geq \tilde{\beta} \tilde{g}_{ij}$ for all $i$ and $j$ {\em (}cf. Proposition \ref{prop_full}{\em )}, or {\em (ii)} $\tilde{\beta} \leq 1/(2n)$ {\em (}cf. Proposition \ref{prop_full_beta}{\em )}.
\end{proposition}

If the synergy graph is a complete graph (i.e., $\tilde{g}_{ij} = 1$ for all $i \neq j$), the eigenvalue of $V$ associated with the eigenvector $\bm{q}_{j} \perp \mathbf{1}_{n}$ is modified as:
\begin{equation}
\omega_{j} = (1 + \beta \lambda_{j})^{-2}(1 - 2\tilde{\beta}n + 2 \beta \lambda_{j}).
\end{equation}
The condition for disclosing $m_{j}^{*} = \bm{q}_{j}' \bm{x}$ is then given by:
\begin{equation}
m_{j}^{*} \in \mathcal{I}(\varphi^{*}) \iff \frac{1}{2n} + \frac{\lambda_{j}}{n} \beta \geq \tilde{\beta}.
\end{equation}

\begin{proposition}\label{prop_asymmetry}
Suppose that $\tilde{g}_{ij} = 1$ for all $i \neq j$ and $\tilde{\beta} > 1/(2n)$. Then:
\begin{enumerate}
\item[{\em i)}] the disclosure of only the average state is optimal if $\tilde{\beta}$ is sufficiently large or $\beta$ is sufficiently small {\em (}cf. Proposition \ref{prop_single_beta}{\em )};
\item[{\em ii)}] the disclosure of full information is optimal if $G$ is connected {\em (}i.e., $\lambda_2 > 0${\em )} and $\beta$ is sufficiently large.
\end{enumerate}
\end{proposition}

This subsection demonstrates that introducing asymmetry in coordination benefits simplifies the interpretation of the parameter $\beta$. When $\tilde{\beta}$ differs from $\beta$, $\beta$ can be directly interpreted as the agents' coordination incentive, while $\tilde{\beta}$ represents the principal's distinct preferences. This separation clarifies the principal's objective and facilitates the analysis of optimal information disclosure policies under varying coordination structures.

\section{Numerical analysis using random graph models}
This section presents a numerical analysis exploring how coordination incentives shape the design of the optimal signal, with a focus on its dimensionality. Subsection \ref{sec:ER_connect} examines how connectivity in the incentive graph enhances transparency, based on the insights from Proposition \ref{prop_morelinks}. Subsection \ref{sec:compare_ER_BA} investigates the role of hubs or central agents, building on the analysis of star graphs in Subsection \ref{sec:example}. Finally, Subsection \ref{sec:approx_min} introduces an approximation method to identify the conditions under which the minimum transparency policy is optimal. Throughout this section, I consider the case where the synergy graph is a complete graph ($\widetilde{G}=K_n$), while the incentive graph is modeled using well-known random graph models. 
 
\subsection{Connectivity and transparency}\label{sec:ER_connect} 

In this subsection, I analyze the optimal signal in networks generated by the \textit{Erd\H{o}s--R\'enyi} (ER) model $G(n,p)$, where the number of nodes $n$ is fixed and each pair is connected with independent probability $p$. The purpose of this analysis is to examine how the link formation probability $p$, or the resulting \textit{expected degree} $(n-1)p$, influences the optimal signal dimension and to what extent the policy responds to the relative importance of adaptation over coordination ($\beta^{-1}$) as well as the randomness of the incentive graph structure.

Specifically, I generate 100 random graphs for each $(n,p)$ in each parameter set and compute the optimal signal dimension by varying $\beta^{-1}$. Figure \ref{fig:ER_dimension} illustrates how the optimal signal dimension changes as the importance of adaptation relative to coordination varies. As established in Proposition \ref{prop_balanced}, when the organization places greater weight on adaptation, the principal discloses more statistics to enable agents to better respond to their local conditions. The vertical dotted line in the figure represents the cutoff beyond which full transparency becomes optimal across all network structures, as presented in Proposition \ref{prop_full_beta}. The error bars indicate the standard deviations, capturing the variability in the optimal signal dimension due to the randomness of the ER model.

\begin{figure}[t]\centering
\includegraphics[width=0.6\textwidth]{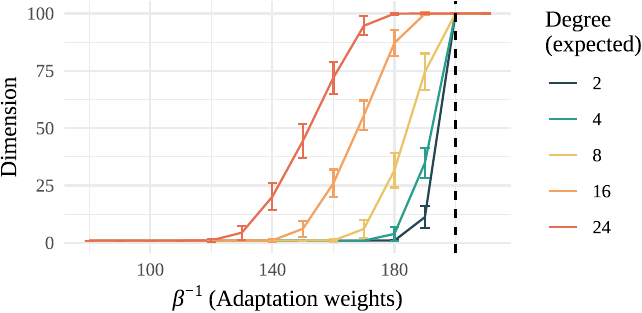}
\caption{Connectivity and optimal signal dimension.}\label{fig:ER_dimension}
\end{figure}

The figure reveals an S-shaped relationship between $\beta^{-1}$ and the signal dimension. At low $\beta^{-1}$, only the average state statistic is disclosed, achieving perfect coordination but poor adaptation. As $\beta^{-1}$ increases, the signal dimension gradually rises, with the principal revealing a few statistics to enhance adaptation while maintaining a high level of coordination. After this gradual phase, the signal dimension increases almost linearly, smoothly balancing adaptation and coordination until $\beta^{-1}$ reaches the cutoff for full revelation.

The expected degree also plays a crucial role in determining the signal dimension. Higher connectivity generally leads to a higher signal dimension, consistent with the theoretical result that more connected networks can support greater transparency. Notably, the range of $\beta^{-1}$ over which the signal dimension remains minimal (disclosing only the average state) reflects the impact of network density. In networks with higher expected degrees, the transition from minimal to full disclosure spans a broader range of $\beta^{-1}$, indicating that the principal can effectively use partial transparency to induce profitable correlations in expectations. By contrast, in less connected networks, these transitions are sharper, and the principal typically chooses either minimal or maximal transparency across most of the $\beta^{-1}$ range.%
\footnote{Note that if $G$ has no links, then the optimal signal is either the average disclosure or full disclosure, depending on whether $\beta^{-1} < 2n$. Thus, the figure confirms that as the network becomes less connected, the range of partial disclosure contracts and eventually disappears.}

\subsection{Comparing optimal signal dimensions in different network models}\label{sec:compare_ER_BA}

In this subsection, I introduce the \textit{Barab\'asi--Albert} (BA) model to examine how the degree distribution of the incentive graph affects the principal's optimal policy. While the ER model, with its uniformly random link formation, produces a Poisson degree distribution, the BA model generates a power-law degree distribution, commonly referred to as a \textit{scale-free network} characterized by highly connected hubs. These structural differences in the network directly influence the complexity of the signal required to balance adaptation and coordination. In particular, the hubs in the BA model may play a role similar to the core agent in the star graph examined in Proposition \ref{prop_star}, enabling the principal to disclose additional information to agents without significantly weakening coordination.

\begin{figure}[t]\centering
	\includegraphics[width=0.85\textwidth]{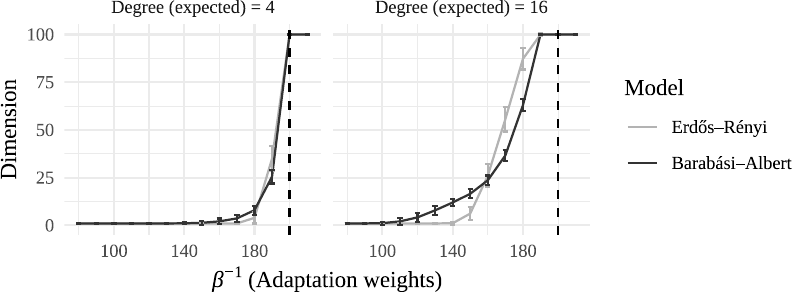}
	\caption{Comparison of optimal signal dimension between ER and BA models ($n=100$).}\label{fig:comparison_ER_PA}
\end{figure}

Figure \ref{fig:comparison_ER_PA} compares the transitions of signal dimensions between the two models while controlling for a constant expected degree. The graph for the ER model (in gray) is identical to the previous figure, showing that the transition from the minimum to maximum dimension occurs linearly within a certain range of $\beta^{-1}$.

In contrast, the \textit{BA model} exhibits a distinct pattern. The signal dimension exceeds the minimum value even at low \(\beta^{-1}\), increasing only gradually over a wide range of $\beta^{-1}$. This indicates that in networks with highly connected hubs, as generated by the BA model, disclosing the local states of these hubs enhances adaptation without significantly reducing coordination. Agents connected to hub nodes naturally align with them, enabling signals that focus on the hubs to sustain co-movement in agents' actions. Meanwhile, additional information improves adaptation by increasing the accuracy of local state estimates. Consequently, the BA model exhibits partial transparency with a more gradual increase in signal dimensions compared to the sharper transitions observed in ER-model networks.

As $\beta^{-1}$ approaches the cutoff for full transparency, both models converge to full transparency (i.e., the dimension equals the number of agents) at similar \(\beta^{-1}\) values, which align closely with the sufficient condition \(\beta^{-1} > 2n\) identified in Proposition \ref{prop_full_beta}. As highlighted in Corollary \ref{corollary_full_beta}, the condition for full revelation is largely driven by the algebraic connectivity (the second-smallest Laplacian eigenvalue). This explains why these two network models, despite their differing degree distributions, produce similar thresholds for full transparency, even though their Laplacian spectral radii (the largest Laplacian eigenvalues) differ substantially.

\subsection{Approximation of minimum transparency condition}\label{sec:approx_min}

This subsection examines the conditions under which the minimum transparency policy---disclosing only the average state---is optimal in large-scale organizations modeled by random graphs. As shown in Section \ref{sec:example}, intermediate disclosure policies depend heavily on the specific details of the network topology. However, acquiring complete knowledge of network structures is often impractical, particularly when coordination incentives arise from social preferences. Identifying the conditions for minimum transparency is therefore crucial for enabling the principal to design robust and effective information policies in complex and uncertain network environments.

Proposition~\ref{prop_balanced} establishes that the minimum transparency policy is optimal when the following condition is satisfied:
\begin{equation}
    \beta > \frac{1}{2(n - \lambda_n)}.
    \label{eq:min_transparency_condition}
\end{equation}
This result highlights the critical role of the Laplacian spectral radius, $\lambda_n$, in determining the threshold for information disclosure.%
\footnote{Appendix \ref{sec:a_approx_full} provides a complementary analysis of the approximation of the full transparency condition, determined by the algebraic connectivity ($\lambda_2$), and its relationship with the minimum degree ($d^{\min}$).}

The spectral graph theory literature provides upper and lower bounds for the maximum Laplacian eigenvalue (\(\lambda_n\), also known as the Laplacian spectral radius). A classical result from \cite{anderson1985eigenvalues} provides an upper bound:
\begin{equation}
    \lambda_n \leq \max\{ d(u) + d(v) : uv \in E \},
\end{equation}
where \(d(u)\) is the degree of vertex \(u\). Further refinements for connected graphs are available in the literature.%
\footnote{For example, \cite{liu2004laplacian} provide
\begin{equation}\nonumber
    \lambda_n \leq \frac{(d^{\max} + d^{\min} - 1) + \sqrt{(d^{\max} + d^{\min} - 1)^2 + 4(4m - 2 d^{\min}(n - 1))}}{2},
\end{equation}
and \cite{shi2007bounds} provides 
$\lambda_n \leq 2 d^{\max} - \frac{2}{n(2Diam + 1)}$, 
where $Diam$ is the diameter and $m$ is the number of edges. More recent works include \cite{shi2009spectral}, \cite{li2010laplacian}, and \cite{jahanbani2021spectrum}, among others.}

While these bounds are theoretically valuable, they often rely on parameters such as the graph diameter or the number of edges, which can be difficult to compute in large networks unless complete information about the network is available. This limitation motivates the use of simpler approximations for \(\lambda_n\) based solely on the maximum degree, as explored in simulations.

\begin{figure}[t]\centering
	\includegraphics[width=0.7\textwidth]{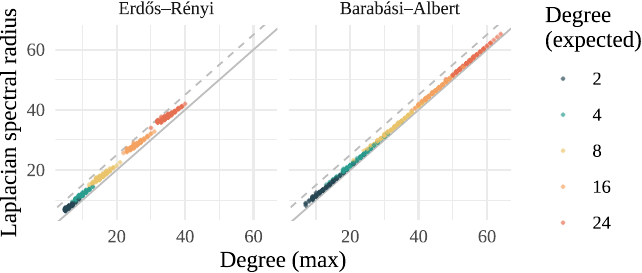}
	\caption{Spectral radius and degree in ER and BA models ($n=100$).}\label{fig:lambda_n}
\end{figure}

Figure~\ref{fig:lambda_n} demonstrates a strong empirical correlation between \(\lambda_n\) and the maximum degree \(d^{\max}\) in networks generated by the ER and BA models of size \(n = 100\).\footnote{For each parameter set, 100 random graphs are generated.} Across varying expected degrees, a nearly linear relationship emerges between \(\lambda_n\) and \(d^{\max}\), regardless of the network generation mechanism. The lower gray line in Figure~\ref{fig:lambda_n} represents the 45-degree line, while the upper dotted gray line is offset by 5 units.%
\footnote{A simple characterization of the lower bound for \(\lambda_n\) is provided by \cite{grone1994laplacian}: \(\lambda_n \geq d^{\max} + 1\).}

For a total of 1,000 samples (500 for each random graph model), all cases satisfy:
\begin{equation}
    d^{\max} + 1 \leq \lambda_n \leq d^{\max} + 5.
\end{equation}
This observation enables the approximation of conditions for minimum transparency using \(d^{\max}\) as a practical proxy for \(\lambda_n\). Such an approximation is particularly advantageous for large networks, where the exact computation of \(\lambda_n\) can be computationally intensive or practically infeasible.%
\footnote{\cite{kim2007ensemble} theoretically examine the approximation of the Laplacian spectral radius for BA graphs as $\lambda_n \simeq d^{\max} + 1$, which aligns with the observation in the right panel of Figure \ref{fig:lambda_n}.}

The approximation provides a practical tool for designing information-sharing policies in organizations. For example, in sparse networks, such as low-degree ER graphs, the maximum degree $d^{\max}$ is small relative to the network size. This results in a larger range of $\beta$ for which minimum transparency is optimal and a narrower range for intermediate disclosure policies. In contrast, in dense or hub-dominated networks, such as high-degree BA graphs, the high $d^{\max}$ increases the threshold for minimum transparency in \eqref{eq:min_transparency_condition}, requiring higher coordination benefits ($\beta$) to justify minimal disclosure.

These findings align with observations in Section~\ref{sec:ER_connect} on connectivity and transparency, and Section~\ref{sec:compare_ER_BA} on comparing network models. Using $d^{\max}$ as a proxy for $\lambda_n$ offers a practical advantage, especially when exact graph parameters, such as the diameter or minimum degree, are unavailable. This approach bridges theoretical results from spectral graph theory with practical considerations for adaptive organizations.

\section{Conclusion}

This paper investigated the optimal information control in an organization to improve adaptation to the external environment and coordination between agents. The results emphasize the importance of sharing the organizational goal, even when the underlying incentive structure does not naturally induce the desired coordination. Full revelation is optimal when the principal aims to restrain excessive coordination. 

When the principal benefits equally from the coordination of any two agents, the condition for full revelation is determined by the algebraic connectivity of the incentive graph, as it reflects the ease of aligning decentralized actions. Conversely, the condition for minimum transparency, where agents are informed only of the overall average of the local states, is governed by the Laplacian spectral radius, which captures the influence of the most connected agents on the network's overall coordination structure.

This analysis assumes that the principal cannot predict or control decentralized communication between agents and, therefore, relies solely on public signals rather than private signals. If decentralized communication can be restricted through organizational design, the principal might achieve further improvements by providing differentiated information to agents. The optimal design of private signals, which involves distinct methodologies, remains an avenue for future research.%
\footnote{\cite{miyashita2024lqg} provide techniques for incorporating private signals into optimal information design with general quadratic payoff settings, including the model studied in this paper.}

\bibliographystyle{ecca}

\appendix
\section*{Appendix}

\section{Basic properties of the Laplacian spectrum}\label{sec:laplacian_a}
This section provides some basic properties of the Laplacian spectrum of a graph that are related to this paper.%
\footnote{All lemmas below are known in spectral graph theory. See for example, \cite{moh:1991book}.}
\begin{lemma}
The Laplacian $L^{K}$ of the complete graph has the eigenvalues such that 
\begin{equation}
0=\lambda_{1}^{K}<\lambda_{2}^{K}=\cdots= \lambda_{n}^{K}=n.	
\end{equation}
\end{lemma}
\begin{proof}
First, it is straightforward to show that $\mathbf{1}_{n}$ is an eigenvector of $L^{K}$ with eigenvalue zero.
Next, consider any vector $\bm{q}=(q^{1}, q^{2}, \dots, q^{n})'$ that is orthogonal to $\mathbf{1}_{n}$.
Since $\bm{q}'\bm{1}_{n}=\sum_{j}q^{j}=0$, 
\begin{equation}\nonumber
L^{K} \bm{q}=
 \begin{bmatrix}
  n-1  &  -1  &   \cdots  &  -1 \\
  -1   & n-1  &   \cdots  &  -1 \\
  \vdots   &  \vdots &   \ddots  &  \vdots  \\
  -1  &  -1  &   \cdots  &  n-1 
 \end{bmatrix}
 \begin{bmatrix}
 q^{1} \\
 q^{2} \\
\vdots  \\
 q^{n} \\
 \end{bmatrix} 
=
\begin{bmatrix} 
(n-1) q^{1}-\sum_{j\neq 1} q^{j} \\
(n-1)q^{2}-\sum_{j\neq 2} q^{j} \\
\vdots \\
(n-1) q^{n}-\sum_{j\neq 2} q^{n} \\
 \end{bmatrix} =n
 \begin{bmatrix}
 q^{1} \\
 q^{2} \\
\vdots  \\
 q^{n} \\
 \end{bmatrix} .
\end{equation}
Thus, $L^{K}$ has eigenvalue $n$ of multiplicity $n-1$ with eigenvector $q_{i} \perp \bm{1}_{n}$ for $i=2, \dots, n$.
\end{proof}

\begin{lemma}\label{lem_laplacian_spectra}
The Laplacian $L$ has the eigenvalues such that 
\begin{equation}
0=\lambda_{1} \leq \lambda_{2} \leq \cdots \leq \lambda_{n} \leq n.
\end{equation}
\end{lemma}
\begin{proof}
The minimum eigenvalue $\lambda_{1}$ of $L$ equals zero because $L$ is positive semidefinite and singular. The corresponding eigenvector is $\bm{q}_{1}=\mathbf{1}_{n}$.
The maximum eigenvalue of $L$ is less than or equal to $n$ because (i) $L^{K}-L$ is positive semidefinite and (ii) $(L^{K}-L)\bm{q}_{i}=(n-\lambda)\bm{q}_{i}$ for any eigenvector $\bm{q}_{i}$ of $L$ that is orthogonal to $\mathbf{1}_{n}$. Thus, $n-\lambda_{i}\geq 0$ for $i=2, \dots, n$.
\end{proof}

\begin{lemma}
The second smallest eigenvalue of Laplacian $L$ is greater than zero ($\lambda_{2}>0$) if and only if $G$ is connected.
\end{lemma}
\begin{proof} 
\textit{Only if part}: Suppose that $G$ is not connected. Without loss of generality, $G$ is divided into two components:
\begin{equation}
G=\begin{bmatrix}
G_{A} & O \\ 
O & G_{B}
\end{bmatrix}.	
\end{equation}
Then, the Laplacian of $G$ has eigenvalue zero of multiplicity at least two because 
\begin{equation}
\begin{bmatrix}
L_{A} & O \\ 
O & L_{B}
\end{bmatrix}
\begin{bmatrix}
\mathbf{1}_{A}\\\mathbf{0}_{B}
\end{bmatrix}
=
\begin{bmatrix}
L_{A} & O \\ 
O & L_{B}
\end{bmatrix}
\begin{bmatrix}
\mathbf{0}_{A}\\\mathbf{1}_{B}
\end{bmatrix}
=
\mathbf{0}_{n}.
\end{equation}

\textit{If part}: Suppose that $G$ is connected. 
Suppose now that $\bm{q}=(q^{1}, \dots, q^{n})'$ is an eigenvector of $L$ with eigenvalue zero (i.e., $L \bm{q}=0$). 
Note that $\bm{q}'L \bm{q}=\frac{1}{2}\sum_{i }\sum_{j}g_{ij}(q^{i}-q^{j})^{2}=0$.
This implies that $q^{i}=q^{j}$ for any pair $(i,j)$ such that $g_{ij}>0$. 
Since $G$ is connected (i.e., any pair $(i,j)$ has a path), $q^{i}=q^{j}$ for all $i$ and $j$. Thus, the eigenvector of $L$ with eigenvalue zero must be proportional to $\mathbf{1}_{n}$.
\end{proof}

\begin{lemma}
The largest eigenvalue of Laplacian $L$ is smaller than $n$ ($\lambda_{n}<n$) if and only if the complement of $G$ is connected.
\end{lemma}
\begin{proof}
Let $K_n$ be a complete graph. Since the complement of $G$ is connected, $L^{K}-L$ has the smallest eigenvalue that is strictly greater than $0$.
For any eigenvector $\bm{q}_{i} \perp \mathbf{1}_{n}$ of $L$, it is the case that $(L^{K}-L)\bm{q}_{i}=(n-\lambda_{i})\bm{q}_{i}$.
Thus, $n-\lambda_{i}>0$ for any $i=2,\dots, n$.
\end{proof}

\section{Derivation}\label{sec:derivation}
\subsection{Equilibrium}\label{sec:derivation_equil}
The first-order condition for the agent $i$'s maximization problem is given by
\begin{equation}
\begin{aligned}\nonumber
\mathbb{E}\left[ 2(a_{i}-x_{i}) +2\beta \sum_{j=1}^{n}g_{ij}(a_{i}-a_{j})|m\right]=&0 ,
\\ 
\left( 1+\beta \sum_{j=1}^{n}g_{ij} \right)a_{i} -\beta \sum_{j=1}^{n}g_{ij}a_{j}=&\hat{x}_{i}
\\
\left( 1+\beta L_{ii} \right)a_{i} +\beta \sum_{j\neq i}L_{ij}a_{j}=&\hat{x}_{i}
\end{aligned}
\end{equation}
where $L_{ii}=\sum_{j=1}^{n}g_{ij}$ and $L_{ij}=-g_{ij}$ for $i \neq j$.
The system of best responses is expressed as 
\begin{equation}
 \begin{bmatrix}
  1+\beta L_{11}&\beta L_{12} &   \cdots  &  \beta L_{1n} \\
  \beta L_{21}   & 1+\beta L_{22}  &   \cdots  & \beta L_{2n} \\
  \vdots   &  \vdots &   \ddots  &  \vdots  \\
  \beta L_{n1}  & \beta L_{n2}  &   \cdots  &  1+\beta L_{nn}
 \end{bmatrix}\begin{bmatrix}
 a_{1}\\a_{2}\\ \vdots \\ a_{n} \end{bmatrix}
 =
 \begin{bmatrix}
 \hat{x}_{1} \\ \hat{x}_{2} \\ \vdots \\ \hat{x}_{n} \end{bmatrix},
\end{equation}
or in matrix form, 
\begin{equation}\nonumber
(I_n+\beta L )\bm{a}=\hat{\bm{x}}.	
\end{equation}
Since $I_n+\beta L$ is strictly diagonally dominant, it must be nonsingular.
Thus, the unique Nash equilibrium is given by $\bm{a}^{*}=(I_n+\beta L)^{-1}\hat{\bm{x}}$. Note that $B = (I_n + \beta L)^{-1}$ is symmetric since $L$ is symmetric.

\subsection{Objective function}\label{sec:derivation_objective}
I derive the principal's objective function \eqref{eq:objective}.
Firstly, I use the following matrix calculus:
\begin{equation}
\begin{aligned}
\bm{a}' G \bm{a} =& 
\begin{bmatrix} a_{1}&\cdots & a_{n} \end{bmatrix}
\begin{bmatrix}\sum_{j}g_{1 j}a_{j} \\ \vdots \\ \sum_{j}g_{n j}a_{j} \end{bmatrix}
=\sum_{i} \sum_{j} g_{ij} a_{i}a_{j},
\\
\bm{a}' D_{G} \bm{a} =& 
\begin{bmatrix} a_{1}&\cdots & a_{n} \end{bmatrix}
\begin{bmatrix}\sum_{j}g_{1 j}a_{1} \\ \vdots \\ \sum_{j}g_{n j}a_{n} \end{bmatrix}
= \sum_{i} \sum_{j} g_{ij} a_{i}^{2}.	
\end{aligned}	
\end{equation}
From these equalities, I have  
\begin{equation}
2 \bm{a}' L \bm{a} = \sum_{j=1}^{n} \sum_{i=1}^{n}g_{ij} (a_{i}-a_{j})^{2}.	
\end{equation}
Then, the principal's objective function given $\bm{a}^{*}=B\hat{\bm{x}}=(I_n + \beta L)^{-1}\hat{\bm{x}}$ reduces to 
\begin{equation}
\begin{aligned}\nonumber
\mathbb{E}[ v(\bm{a}^{*}, \bm{x})] 
=&
 -\sum_{i=1}^{n} \mathbb{E}\left[ a_{i}^{2}-2a_{i}x_{i}+x_{i}^{2} \right] - \beta \sum_{i=1}^{n} \sum_{j=1}^{n} \tilde{g}_{ij} \mathbb{E}\left[ (a_{i}-a_{j})^{2}\right]
 \\
=&
 -\mathbb{E}[\bm{a}'\bm{a}-2\bm{a}'\bm{x}+\bm{x}'\bm{x}] - 2  \beta  \mathbb{E}\left[\bm{a}'\widetilde{L}\bm{a} \right]
 \\
=&
 -\mathbb{E}[\hat{\bm{x}}'B'B\hat{\bm{x}}-2\hat{\bm{x}}'B'\hat{\bm{x}}+\bm{x}'\bm{x}] - 2 \beta 
\mathbb{E}\left[\hat{\bm{x}}'B'\widetilde{L}B \hat{\bm{x}} \right] 
\\
=&
 -\mathbb{E}[\hat{\bm{x}}'B'B\hat{\bm{x}}-2\hat{\bm{x}}'B'(I_n + \beta L)B\hat{\bm{x}}+\bm{x}'\bm{x}] - 2 \beta 
\mathbb{E}\left[\hat{\bm{x}}'B'\widetilde{L}B\hat{\bm{x}}\right] 
\\
=& -\mathbb{E}[\bm{x}'\bm{x}] + \mathbb{E}[\hat{\bm{x}}' B' (I_n - 2\beta (\widetilde{L}-L ))B\hat{\bm{x}}].
\end{aligned}
\end{equation}
Note that from the law of iterated expectations, I can replace $\mathbb{E}[\hat{\bm{x}}'B\bm{x}]$ by $\mathbb{E}[\mathbb{E}[\hat{\bm{x}}'B\bm{x}|m]]=\mathbb{E}[\hat{\bm{x}}'B\hat{\bm{x}}]$.

\section{Proofs}\label{sec:proof}
\begin{proof}[Proof of Proposition \ref{prop_avg}]
From Theorem \ref{theorem}, it suffices to show that $V \mathbf{1}_{n} = \omega_{1} \mathbf{1}_{n}$ for some $\omega_{1}>0$.
First, Laplacian matrices $L$ and $\widetilde{L}$ have an eigenvector of $q_{1}=\mathbf{1}_{n}$ with eigenvalue $\lambda_{1}=0$.
Next, for any eigenvector $\bm{q}_{i}$ of $L$, $B \bm{q}_{i}=(I+\beta L)^{-1}\bm{q}_{i}= (1+\beta \lambda_{i})^{-1}\bm{q}_{i}$ where $\lambda_{i}$ is the eigenvalue of $L$ associated with $\bm{q}_{i}$.
Then, it is straightforward to show that $V \bm{1}_{n}= \bm{1}_{n}$.
\end{proof}

\begin{proof}[Proof of Proposition \ref{prop_full}]
I will show that $V$ is positive definite if $G \geq \widetilde{G}$ in terms of component-wise inequality.
First, $G-\widetilde{G}$ is an adjacency matrix when $g_{ij} \geq \tilde{g}_{ij}$ for all $i$ and $j$.
Then, $L-\widetilde{L}$ is diagonally dominant since it is the Laplacian matrix of $G-\widetilde{G}$.
Second, $I_n+2 \beta (L-\widetilde{L})$ is symmetric and strictly diagonally dominant. 
Hence, $V=B'(I_n-2\beta (\widetilde{L}-L))B$ is positive definite.
\end{proof}

\begin{proof}[Proof of Proposition \ref{prop_full_beta}]
I will show that $V$ is positive semidefinite if $\beta$ is sufficiently close to zero.
Since $\widetilde{L}$ is symmetric, I have the eigenvalue decomposition: $\widetilde{L} =\widetilde{Q}\widetilde{\Lambda} \widetilde{Q}'$. 
Then,%
\begin{equation}\nonumber
I_n - 2\beta \widetilde{L} + 2 \beta L 
= \widetilde{Q} [ I_n - 2 \beta \widetilde{\Lambda}]\widetilde{Q}' + 2 \beta L.	
\end{equation}
Since every eigenvalue of the Laplacian matrix is not greater than $n$, $I_n - 2\beta \widetilde{\Lambda}$ is positive semidefinite if $1-2\beta n \geq 0$.
Hence, if $\beta \leq 1/(2n)$, then $V=B'(I_n - 2\beta (\widetilde{L}-L))B$ is positive semidefinite.
\end{proof}

\begin{proof}[Proof of Proposition \ref{prop_morelinks}]
To prove the proposition, I utilize Weyl's monotonicity theorem, which states that for any symmetric matrix $A$ and positive semidefinite matrix $E$, I have $\lambda_k(A) \leq \lambda_k(A+E)$ for $k=1,2,, \dots, n$.%
\footnote{This theorem is an immediate corollary of the Courant-Fischer theorem.}

Fix $G$ and $G'$ with $G' \geq G$. I will show that the payoff matrix for $G'$, denoted by $V_{G'}$, has larger eigenvalues than $V_G$ has. 
\begin{equation}\nonumber
	\begin{aligned}
		V_{G'} &= B'[I_n -2\beta(\widetilde{L} - L_{G'})]B
		\\ &= V + 2 \beta B'(L_{G'}-L_G)B,
	\end{aligned}
\end{equation}
where $L_G$ and $L_{G'}$ are the Laplacian matrices for $G$ and $G'$, respectively.
When $G' \geq G$, the matrix $G'-G$ corresponds to the adjacency matrix for the graph consisting of the newly added links, and $L_{G'}-L_G$ is the corresponding Laplacian matrix of $G'-G$. Since any Laplacian matrix is positive semidefinite, so is $B'(L_{G'}-L_G)B$. Then, I can apply Weyl's monotonicity theorem, completing the proof.
\end{proof}

\begin{proof}[Proof of Proposition \ref{prop_beta}]
Consider the eigendecomposition	of $\Delta \equiv \widetilde{L}-L$ denoted by $Q_\Delta \Lambda_\Delta Q_\Delta'$. Then $V$ is written as 
\begin{equation}\nonumber
V = B' Q_\Delta [I_n - 2\beta \Lambda_\Delta]Q_\Delta'B.
\end{equation}
Note that the element of $\Lambda_\Delta$ may be negative because $\widetilde{L}-L$ is not necessarily positive semidefinite.
Then, the number of positive eigenvalues cannot increase when $\beta$ increases.
\end{proof}

\begin{proof}[Proof of Proposition \ref{prop_balanced}]
Let $\bm{q}_{j}$ be the eigenvector of $L$ with eigenvalue $\lambda_{j}$.
I will show that if $\widetilde{L}=L^{K}$ (the Laplacian of a complete graph), then $V\bm{q}_{j}=\omega_{j} \bm{q}_{j}$ for some $\omega_{j} \in \mathbb{R}$.
Since $\bm{q}_{1}=\mathbf{1}_{n}$ is an eigenvector of $V$, I focus on $\bm{q}_{j} \perp \mathbf{1}_{n}$.

First, $(L^{K}-L)\bm{q}_{j}=(n-\lambda_{j}) \bm{q}_{j}$ since $L^{K}\bm{q}_{j} = n \bm{q}_{j}$ for any $\bm{q}_{j}\perp \mathbf{1}_{n}$. 
Second, $(I_n+\beta L)^{-1}\bm{q}_{j} = (1+\beta \lambda_{j})^{-1} \bm{q}_{j}$.
Then, $V\bm{q}_{j}=\omega_{j} \bm{q}_{j}$ where 
\begin{equation}\label{eq:omega}
\omega_{j} \equiv (1+\beta \lambda_{j})^{-2} (1-2\beta (n-\lambda_{j})) .
\end{equation}
The eigenvalue $\omega_{j}$ is nonnegative if and only if $1\geq2\beta (n -\lambda_{j})$.
\end{proof}

\begin{proof}[Proof of Proposition \ref{prop_single_beta}]
Suppose that $\widetilde{G}$ is the adjacency matrix of a complete graph and that the complement of $G$ is connected.
Then the second smallest eigenvalue of the Laplacian of the adjacency matrix $\widetilde{G}-G$ is strictly greater than zero.
This implies that the largest eigenvalue $\lambda_{n}$ of $L$ is strictly less than $n$.
Hence, the condition $\beta >1/(2(n-\lambda_{n}))$ is sufficient to ensure that $\omega_2, \dots, \omega_n<0$. In this case, $\omega_1>0$ is the only nonnegative eigenvalue of $V$.
\end{proof}

\begin{proof}[Proof of Proposition \ref{prop_gain}]
Let $\omega_{1}, \dots, \omega_{n}$ be the eigenvalues of $V$.
Under the optimal signal, the gain from information design $\mathbb{E}[\hat{\bm{x}}'V\hat{\bm{x}}]-\bm{\mu}'V'\bm{\mu}$ is proportional to the sum of positive eigenvalues of $V$.
Formally, 
\begin{equation}
\mathbb{E}[\hat{\bm{x}}'V\hat{\bm{x}}] -\bm{\mu}'V'\bm{\mu} 
 = \sigma^{2} \sum_{\{ i: \omega_{i}>0 \} }  \omega_{i}.	
\end{equation}
The eigenvalue of $V$ associated with eigenvector $\bm{q}_{1}=\mathbf{1}_{n}$ is 1.
Other eigenvalues $\omega_{j}$ of $V$ are given by \eqref{eq:omega}.
If $\lambda_{j}=n$, then $\omega_{j} \to 0$ as $\beta \to \infty$. If $\lambda_{j}<n$, then $\omega_{j}<0$ for sufficiently high $\beta$.
Hence, the sum of positive eigenvalues of $V$ converges to 1.
\end{proof}

\begin{proof}[Proof of Proposition \ref{prop_star}]
It is straightforward to verify that the Laplacian spectrum of the star graph and the associated eigenvectors are given by 
\begin{center}
\begin{tabular}{ccccc}
$\lambda_j$:&$0^{(1)}$ & $1^{(n-2)}$ & $ n^{(1)}$
\\
$\bm{q}_j$:&$\begin{bmatrix}
1 \\ \mathbf{1}_{n-1}
\end{bmatrix}$
&
$\begin{bmatrix}
0 \\ \bm{y}_{j}
\end{bmatrix}$
&
$\begin{bmatrix}
(n-1) \\ - \mathbf{1}_{n-1}
\end{bmatrix}$
\end{tabular} 
\end{center}
where $\bm{y}_{j}$ is a vector of length $n-1$ that is orthogonal to $\bm{1}_{n-1}$.%
\footnote{Here, the superscripts with parentheses indicate the multiplicities.}
\end{proof}

\begin{proof}[Proof of Proposition \ref{prop_informativeness}]
Let $\bar{Q} = \begin{bmatrix} \bm{q}_1&\bm{q}_n&\cdots &\bm{q}_r \end{bmatrix}$ be the matrix consisting of the relevant eigenvectors used for weighting the disclosed statistics.
Then the disclosed statistics are $\bm{m}^* = \bar{Q}' \bm{x}$. The conditional expectation of $\bm{x}$ given $\mathbf{m}^*$ is
\begin{equation}\nonumber
\mathbb{E}[\bm{x}|\bm{m}^*] = \bar{Q}(\bar{Q}'\bar{Q})^{-1} \bm{m}^*,
\end{equation}
and its variance matrix is 
\begin{equation}\label{eq:var_matrix}
\text{var}(\hat{\bm{x}}) = \sigma^2 \bar{Q}(\bar{Q}'\bar{Q})^{-1} \bar{Q}'.
\end{equation}
When I apply the standard orthonormality conditions for the eigenvectors, namely $\bm{q}_i'\bm{q}_i =1$ for each $i$ and $\bm{q}_i' \bm{q}_j = 0$ for $i \neq j$, \eqref{eq:var_matrix} simplifies to
\begin{equation}\label{eq:var_matrix_reduced}
\text{var}(\hat{\bm{x}}) = \sigma^2 \sum_{j \in \mathcal{I}^*} \bm{q}_j \bm{q}_j'.
\end{equation}
The $i$-th diagonal element of \eqref{eq:var_matrix_reduced} is precisely the sum of the squared $i$-th components of those eigenvectors.
\end{proof}

\begin{proof}[Proof of Proposition \ref{prop_plusone}]
To simplify the notation, consider the plus-one policy targeting agent $1$. Under this policy, the variance matrix of the posterior expectations is given by:
\begin{equation}\nonumber
S = \sigma^2
\begin{bmatrix}
1 & \bm{0}_{n-1}'
\\
\bm{0}_{n-1} & \frac{1}{n-1} \bm{1}_{n-1} \bm{1}_{n-1}'
\end{bmatrix}.
\end{equation}

Let $V = Q \Omega Q'$ be the eigendecomposition of $V$, where $Q$ is the matrix containing the Laplacian eigenvectors as column vectors, and $\Omega = \mathrm{diag}(\omega_1, \omega_2, \dots, \omega_n)$ is the diagonal matrix with the eigenvalues given by \eqref{eq:omega}. The organizational payoff gain compared to the minimal transparency policy is:
\begin{equation}
\begin{aligned}\nonumber
\textrm{tr}(VS)- \sigma^2 &=~ \textrm{tr}(Q\Omega Q' S)- \sigma^2 
\\ 
&=~ \textrm{tr}(\Omega Q'SQ)- \sigma^2 
\\
&= \sum_{j=1}^{n} \omega_j \left[Q'SQ \right]_{jj}- \sigma^2,
\end{aligned}	
\end{equation}
where $\left[Q'SQ \right]_{jj}$ represents the $j$-th diagonal element of matrix $Q'SQ$.
Recall that $\omega_1=1$ and $\bm{q}_1 = \frac{1}{\sqrt{n}}\bm{1}_n$, and also that for $j \neq 1$, $\omega_j$ is given by \eqref{eq:omega}.

Next, compute $Q' S Q$:
\begin{equation}
\begin{aligned}\nonumber
Q'SQ &=  \sigma^2
\begin{bmatrix}
\bm{q}_1' \\ \bm{q}_2' \\ \vdots \\ \bm{q}_n'
\end{bmatrix}
\begin{bmatrix}
1 & \bm{0}_{n-1}'
\\
\bm{0}_{n-1} & \frac{1}{n-1} \bm{1}_{n-1} \bm{1}_{n-1}'	
\end{bmatrix}
\begin{bmatrix}
\bm{q}_1&\bm{q}_2 & \cdots & \bm{q}_n
\end{bmatrix}
\\
&= \sigma^2
\begin{bmatrix}
q_1(1)& \frac{1}{n-1}\sum_{j\neq 1}q_1(j) & \cdots & \frac{1}{n-1}\sum_{j\neq 1}q_1(j) \\
q_2(1)& \frac{1}{n-1}\sum_{j\neq 1}q_2(j) & \cdots & \frac{1}{n-1}\sum_{j\neq 1}q_2(j) \\
\vdots & \vdots & \ddots & \vdots \\
q_n(1)& \frac{1}{n-1}\sum_{j\neq 1}q_n(j) & \cdots & \frac{1}{n-1}\sum_{j\neq 1}q_n(j) \\
\end{bmatrix}
\begin{bmatrix}
\bm{q}_1&\bm{q}_2 & \cdots & \bm{q}_n
\end{bmatrix}
\\
&= \sigma^2
\begin{bmatrix}
\frac{1}{\sqrt{n}}& \frac{1}{n-1}(\frac{n}{\sqrt{n}}-\frac{1}{\sqrt{n}}) & \cdots & \frac{1}{n-1}(\frac{n}{\sqrt{n}}-\frac{1}{\sqrt{n}}) \\
q_2(1)& \frac{1}{n-1}(-q_2(1)) & \cdots & \frac{1}{n-1}(-q_2(1)) \\
\vdots & \vdots & \ddots & \vdots \\
q_n(1)& \frac{1}{n-1}(-q_n(1)) & \cdots & \frac{1}{n-1}(-q_n(1)) \\
\end{bmatrix}
\begin{bmatrix}
\frac{1}{\sqrt{n}}\bm{1}_n &\bm{q}_2 & \cdots & \bm{q}_n
\end{bmatrix}
\\
&=  \sigma^2
\begin{bmatrix}
1 & 0 & 0&\cdots & 0
\\
0 & \frac{n}{n-1}[q_2(1)]^2 & \frac{n}{n-1}q_2(1)q_3(1) & \cdots & \frac{n}{n-1}q_2(1)q_n(1)
\\
\vdots & \vdots &\vdots &\ddots  & \vdots
\\
0 & \frac{n}{n-1}q_n(1)q_2(1) & \frac{n}{n-1}q_n(1)q_3(1)& \cdots  & \frac{n}{n-1}[q_n(1)]^2
\end{bmatrix}.
\end{aligned}
\end{equation}
The trace simplifies to:
\begin{equation}\nonumber
\mathrm{tr}(VS) = \sigma^2 + \frac{\sigma^2 n}{n-1} \sum_{j=2}^{n} \bigl[q_j(1)\bigr]^2.
\end{equation}
This completes the proof.
\end{proof}

\begin{proof}[Proof of Lemma \ref{lem_eigen_rhov}]
First, consider $\bm{q}_{1}=\mathbf{1}_{n}$. Recall that $V \bm{q}_{1}= \mathbf{1}_{n}$.
Then, $\Sigma^{\frac{1}{2}}V\Sigma^{\frac{1}{2}}\bm{1}_{n}
=\sigma^{2} (1-\rho+n\rho )\bm{1}_{n}$.
Next, consider $\bm{q}_{j}$ such that $V \bm{q}_{j}=\omega_{j} \bm{q}_{j}$ for some $\omega_{j}$ and $\mathbf{1}_{n}'\bm{q}_{j}=0$.
Then, $\Sigma^{\frac{1}{2}}V\Sigma^{\frac{1}{2}}\bm{q}_{j}
=\sigma^{2} (1-\rho)\omega_{j} \bm{q}_{j}$.
\end{proof}

\begin{proof}[Proof of Proposition \ref{prop_correlation}]
Let $\bm{q}_{j}$ be the $j$--th eigenvector of $V=(I_{n}+\beta L)^{-1}(I_{n}-2\beta (\widetilde{L}-L ))(I_{n}+\beta L)^{-1}$ with associated eigenvalue $\omega_{j}$.
When $\rho=0$, the optimal signal informs the agents of multiple statistics $(m_{1}^{*}, \dots, m_{r}^{*})$ constructed by $m_{j}^{*}=\bm{q}_{j}' \bm{x}$ when $\omega_{1}\geq \cdots \omega_{r} \geq 0 > \omega_{r+1}\geq \cdots \geq \omega_{n}$.

From Lemma \ref{lem_eigen_rhov}, the eigenvectors of $\Sigma^{\frac{1}{2}}V\Sigma^{\frac{1}{2}}$ are identical to those of $V$, and the sign of each eigenvalue of $\Sigma^{\frac{1}{2}}V\Sigma^{\frac{1}{2}}$ coincides with that of the corresponding eigenvalue of $V$.
Thus the statistics for disclosure are given by  
\begin{equation}
\widetilde{m}_{1}^{*}= \bm{q}_{1}' \Sigma^{\frac{1}{2}} \bm{x}
=\sigma  \sqrt{1-\rho +n\rho}~m_{1}^{*}	
\end{equation}
and for $j=2, \dots, n$,
\begin{equation}
\widetilde{m}_{j}^{*}=\bm{q}_{j}' \Sigma^{\frac{1}{2}}\bm{x}
= \sigma \sqrt{1-\rho}~m_{j}^{*}.	
\end{equation}
Since $m_{j}^{*}$ and its monotone transformation $\widetilde{m}_{j}^{*}$ convey the same information to the agents, the optimal signal for $\rho \in (0,1)$ informs the agents of $m_{j}^{*}$ if and only if $m_{j}^{*}$ is revealed under $\rho=0$.
\end{proof}

\section{Approximation of Full Transparency Condition}\label{sec:a_approx_full}

This section examines the approximation of algebraic connectivity (\(\lambda_2\)), the second smallest Laplacian eigenvalue, in terms of the minimum degree (\(d^{\min}\)) for random graphs. Algebraic connectivity is a key parameter that reflects the overall connectivity of a graph, with \(\lambda_2 = 0\) if and only if the graph is disconnected. By definition, it holds that \(0 < \lambda_2\) for any connected graph.

Figure~\ref{fig:lambda_2} illustrates the relationship between \(\lambda_2\) and the minimum degree (\(d^{\min}\)) for Erd\H{o}s--R\'enyi (ER) and Barab\'asi--Albert (BA) models with \(n = 100\). Across varying expected degrees, we observe that:
\begin{equation}
    d^{\min} - 5 \leq \lambda_2 \leq d^{\min}.
\end{equation}
This empirical finding suggests that \(\lambda_2\) can be effectively approximated using \(d^{\min}\), making it a practical proxy for algebraic connectivity in large-scale graphs.

While the bounds are not as tight as those for \(\lambda_n\) discussed earlier, they provide a useful framework for analyzing the structural properties of random graphs, particularly when precise computation of \(\lambda_2\) is infeasible.

\begin{figure}[t]\centering
	\includegraphics[width=0.7\textwidth]{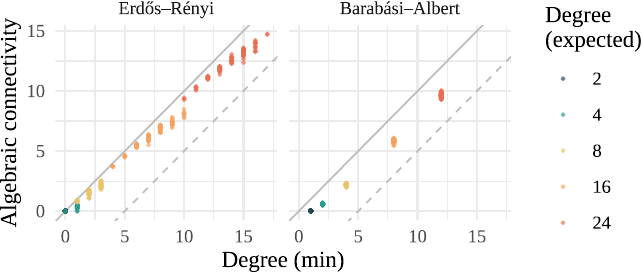}
	\caption{Algebraic connectivity and degree in ER and BA models ($n=100$).}\label{fig:lambda_2}
\end{figure}

\section{The Laplacian spectrum of special graphs}\label{sec:special}

This section provides the Laplacian spectrum of special graphs examined in Section \ref{sec:example}. The Laplacian spectrum consists of the eigenvalues of the Laplacian matrix associated with an undirected graph, and these eigenvalues characterize various properties of the graph. In the notation used here:

\begin{itemize}
    \item \textbf{Eigenvalues and Their Multiplicities:} The eigenvalues are listed with their multiplicities denoted by superscripts. For example, \( 0^{1} \) means that the eigenvalue \( 0 \) has a multiplicity of 1.
    
    \item \textbf{Eigenvectors:} The corresponding eigenvectors are presented below the eigenvalues. The eigenvectors are denoted by vectors, such as \( \mathbf{1}_{n} \), which represents an \( n \)-dimensional vector of ones. The notation \( \bm{y}_{n} \) represents a generic eigenvector orthogonal to \( \mathbf{1}_{n} \), i.e., \( \bm{y}_{n} \perp \mathbf{1}_{n} \).
\end{itemize}

Given these notations, the Laplacian eigenvalues and eigenvectors for some specific graph examples are provided, for instance, in \cite{bh:2012book}:

\begin{enumerate}
\item \textbf{Complete graph} $K_{n}$
\begin{center}
\begin{tabular}{cccc}
$0^{1}$ & $ n^{n-1}$
\\
$\begin{bmatrix}
\mathbf{1}_{n}
\end{bmatrix}$ & $\begin{bmatrix}
\bm{y}_{n}
\end{bmatrix}$
\end{tabular} 
\end{center}
where $\bm{y}_{n} \perp \mathbf{1}_{n}$.
\item \textbf{Star graph} $S_{n}$
\begin{center}
\begin{tabular}{cccc}
$0^{1}$ & $1^{n-1}$ & $ n^{1}$
\\
$\begin{bmatrix}
1 \\ \mathbf{1}_{n-1}
\end{bmatrix}$
&
$\begin{bmatrix}
0 \\ \bm{y}_{n-1}
\end{bmatrix}$
&
$\begin{bmatrix}
(n-1) \\ - \mathbf{1}_{n-1}
\end{bmatrix}$
\end{tabular} 
\end{center}
where $\bm{y}_{n-1} \perp \mathbf{1}_{n-1}$.
\item \textbf{Path graph} $P_{n}$
\begin{center}
\begin{tabular}{cccccc}
$0^{1}$ 
& $2\left( 1- \cos \left(\frac{(j-1)\pi}{n} \right) \right)$
\\
$\begin{bmatrix}
1 \\ 1 \\ \vdots \\ 1
\end{bmatrix}$
& 
$\begin{bmatrix}
\cos \left( \frac{(j-1)\pi}{n}\left(1-\frac{1}{2} \right) \right) 
\\ \cos \left( \frac{(j-1)\pi}{n}\left(2-\frac{1}{2} \right) \right) 
\\ \vdots \\
\cos \left( \frac{(j-1)\pi}{n}\left(n-\frac{1}{2} \right) \right) 
\end{bmatrix}$
\end{tabular}
\end{center}
for $j=2, \dots, n$.

\item \textbf{Complete bipartite graph} $K_{m,\bar{m}}$ ($m \leq \bar{m}$)
\begin{center}
\begin{tabular}{cccc}
$0^{1}$ & $m^{\bar{m}-1}$ & $\bar{m}^{m-1}$ & $(m+\bar{m})^{1}$
\\
$\begin{bmatrix}
\mathbf{1}_{m} \\ \mathbf{1}_{\bar{m}}
\end{bmatrix}$
&
$\begin{bmatrix}
\mathbf{0}_{m} \\ \bm{y}_{\bar{m}}
\end{bmatrix}$
&
$\begin{bmatrix}
\bm{y}_{m}\\ \mathbf{0}_{\bar{m}}
\end{bmatrix}$
&
$\begin{bmatrix}
\bar{m}\mathbf{1}_{m} \\ -m\mathbf{1}_{\bar{m}}
\end{bmatrix}$
\end{tabular} 
\end{center}
where $\bm{y}_{m} \perp \mathbf{1}_{m}$ and $\bm{y}_{\bar{m}} \perp \mathbf{1}_{\bar{m}}$.
\end{enumerate}

\subsection{Additional examples}
\subsubsection{Path graphs}
A path graph $P_{n}$ is a connected graph whose maximum degree is 2.
For simplicity, suppose that vertex $i \in \{2, \dots, n-1\}$ is adjacent to vertices $i-1$ and $i+1$.%
\footnote{Thus, vertices 1 and $n$ have a degree of 1.}
The Laplacian spectrum of the path graph is such that for $j=1, \dots, n$,
\begin{equation}
\begin{aligned}\nonumber
\lambda_{j}=2\left( 1-\cos \left( \frac{(j-1)\pi}{n} \right) \right).
\end{aligned}
\end{equation}
Thus, the eigenvalues are distinct.

\begin{proposition}\label{prop_path}
Suppose that $\widetilde{G}=K_n$ and $G=P_{n}$.
The number of statistics to be disclosed varies from one (the average state) to $n$ (full revelation) as $\beta$ decreases.
\end{proposition}
\begin{proof}%[Proof of Proposition \ref{prop_path}]
In spectral graph theory, the Laplacian spectrum of the path graph and the associated eigenvectors are as presented in Appendix \ref{sec:special}.
\end{proof}

For $n=4$, the Laplacian spectrum of the path graph is $\langle 0, 2-\sqrt{2}, 2, 2+\sqrt{2} \rangle$ with the associated eigenvectors $\langle (1,1,1,1), (a,c,-c,-a), (b,-b,-b,b), (c,-a,a,-c) \rangle$, where $a=\cos (\pi/8)$, $b=\cos (2\pi/8)$, and $c=\cos (3\pi/8)$.

\subsubsection{Complete bipartite graphs}

A complete bipartite graph $K_{m,\bar{m}}$ is a graph where the vertex set can be divided into two disjoint subsets, one with $m$ vertices and the other with $\bar{m}$ vertices, such that every vertex in the first subset is connected to every vertex in the second subset, and there are no edges within each subset. Without loss of generality, suppose that $m \leq \bar{m}$.

When the structure of the equilibrium coordination links corresponds to the complete bipartite graph, the optimal signal will depend on the relative size of the coordination benefit $\beta$ compared to the eigenvalues of the graph's Laplacian matrix.
\begin{proposition}\label{prop_bipartite}
Suppose that $\widetilde{G}=K_n$ and $G=K_{m,\bar{m}}$, such that the coordination structure is complete bipartite. If $\beta > \frac{1}{2m}$, the optimal signal consists of two statistics, each representing the average state of one subset. If $\beta \in \left(\frac{1}{2\bar{m}},\frac{1}{2m}\right]$, then the optimal signal informs the agents of all local states for the agents in the smaller subset, so that the signal has $m+1$ dimensions. If $\beta < \frac{1}{2m}$, full revelation is optimal.
\end{proposition}
\begin{proof}%[Proof of Proposition \ref{prop_bipartite}]
Let $m$ and $\bar{m}$ be the size of each subset of the complete bipartite graph. Applying Proposition \ref{prop_balanced}, the criteria for $\beta$ is given by $(2\bar{m})^{-1}$ and $(2m)^{-1}$, which correspond to eigenvalues $\lambda=m$ and $\lambda=\bar{m}$, respectively. The eigenvectors are presented in Appendix \ref{sec:special}.
\end{proof}

For the complete bipartite graph, as $\beta$ increases, the optimal signal transitions from full disclosure to partial disclosure, where the local states of vertices in one subset are disclosed while the local states in the other subset are averaged. This pattern is driven by the need to facilitate coordination between the two distinct subsets of the graph, ensuring that the coordination incentives are aligned with the organizational objectives.

For example, for $m=2$ and $\bar{m}=3$, the Laplacian spectrum for this complete bipartite graph is $\langle 0, 2^2, 3, 5 \rangle$ with the associated eigenvectors $\langle \mathbf{1}_{5}, (\textbf{0}_2',\bm{y}_{3}')', (\bm{y}_{2}',\textbf{0}_3')', (3\mathbf{1}_{2}', -2\mathbf{1}_{3}')' \rangle$, where $\bm{y}_{2} \perp \mathbf{1}_{2}$ and $\bm{y}_{3} \perp \mathbf{1}_{3}$.

\end{document}